\documentclass[letterpaper, 10 pt, conference]{ieeeconf}%

\IEEEoverridecommandlockouts

\usepackage{cite}
\usepackage{amsmath,amssymb,amsfonts}
\usepackage{graphicx}
\usepackage{textcomp}
\usepackage{xcolor}
\usepackage{cuted}
\usepackage{lipsum}
\usepackage{mathtools}
\usepackage{stfloats}
\usepackage{bbm}
\usepackage{multirow}
\usepackage{graphicx}
\usepackage{subcaption}

\usepackage{algorithm,tabularx}
\usepackage{algpseudocode}

\let\labelindent\relax
\usepackage{enumitem}

\newcommand{\R}{\mathbb{R}}
\newcommand{\N}{\mathbb{N}}
\newcommand{\T}{\top}
\newcommand{\I}{\mathbf{I}}
\newcommand{\0}{\mathbf{0}}

\newcommand{\X}{\mathcal{X}}
\newcommand{\diag}{\text{diag}}
\newcommand{\tsup}[1]{\textsuperscript{#1}}
\newcommand{\mb}[1]{\mathbf{#1}}

\newcommand{\bm}[1]{\begin{bmatrix}#1\end{bmatrix}}
\newcommand{\e}{\mathrm{e}}
\newcommand{\leqew}{\leq\hspace{-1ex}\raisebox{0.2ex}{$<$}\,}

\makeatletter
\newcommand{\multiline}[1]{  \begin{tabularx}{\dimexpr\linewidth-\ALG@thistlm}[t]{@{}X@{}}
#1
\end{tabularx}
}
\makeatother
\algdef{SE}[DOWHILE]{Do}{doWhile}{\algorithmicdo}[1]{\algorithmicwhile\ #1}


\def\BibTeX{{\rm B\kern-.05em{\sc i\kern-.025em b}\kern-.08em
    T\kern-.1667em\lower.7ex\hbox{E}\kern-.125emX}}
\DeclareMathAlphabet\mathbfcal{OMS}{cmsy}{b}{n}

\newtheorem{assumption}{\textbf{Assumption}}
\newtheorem{theorem}{\textbf{Theorem}}
\newtheorem{corollary}{\textbf{Corollary}}
\newtheorem{lemma}{\textbf{Lemma}}
\newtheorem{proposition}{\textbf{Proposition}}
\newtheorem{remark}{\textbf{Remark}}
\newtheorem{definition}{\textbf{Definition}}
\newtheorem{problem}{\textbf{Problem}}

\usepackage{hyperref}

\begin{document}

\title{
{\LARGE \textbf{
Hierarchical Analysis and Control of Epidemic Spreading over Networks using Dissipativity and Mesh Stability
}}}

% \author{\IEEEauthorblockN{Mohammad Javad Najafirad}
% \IEEEauthorblockA{\textit{Dept. of Electrical Engineering} \\
% \textit{Stevens Institute of Technology}\\
% New Jersey, USA \\
% mnajafir@stevens.edu}
% \and
% \IEEEauthorblockN{Shirantha Welikala}
% \IEEEauthorblockA{\textit{Dept. of Electrical Engineering} \\
% \textit{Stevens Institute of Technology}\\
% New Jersey, USA \\
% swelikal@stevens.edu}
% }

\author{Shirantha Welikala and Hai Lin and Panos J. Antsaklis 
% \thanks{$^{\star}$Supported in part by.... } 
\thanks{The support of the National Science Foundation (Grant No. CNS-1830335, IIS-2007949) is gratefully acknowledged.}
\thanks{Shirantha Welikala is with the Department of Electrical and Computer Engineering, Stevens Institute of Technology, Hoboken, NJ \texttt{{\small swelikal@stevens.edu}}. Hai Lin and Panos J. Antsaklis are with the Department of Electrical Engineering, University of Notre Dame, Notre Dame, IN \texttt{{\small \{hlin1,pantsaklis\}@nd.edu}}.}}

\maketitle

%\pagenumbering{arabic}
%\thispagestyle{plain}
%\pagestyle{plain}

\begin{abstract}
Analyzing and controlling spreading processes are challenging problems due to the involved non-linear node (subsystem) dynamics, unknown disturbances, complex interconnections, and the large-scale and multi-level nature of the problems. The dissipativity concept provides a practical framework for addressing such concerns, thanks to the energy-based representation it offers for subsystems and the compositional properties it provides for the analysis and control of interconnected (networked) systems comprised of such subsystems. Therefore, in this paper, we utilize the dissipativity concept to analyze and control a spreading process that occurs over a hierarchy of nodes, groups, and a network (i.e., a spreading network). We start by generalizing several existing results on dissipativity-based topology design for networked systems. Next, we model the considered spreading network as a networked system and establish the dissipativity properties of its nodes. The generalized topology design method is then applied at multiple levels of the considered spreading network to formulate its analysis and control problems as Linear Matrix Inequality (LMI) problems. We identify and enforce localized necessary conditions to support the feasibility of the LMI problem solved at each subsequent hierarchical level of the spreading network. Consequently, the proposed method does not involve iterative multi-level optimization stages that are computationally inefficient. The proposed control solution ensures that the spreading network is not only stable but also dissipative and mesh-stable (i.e., it also minimizes the impact and propagation of disturbances affecting the spreading network). Compared to conventional methods, such as threshold pruning and high-degree edge removal, our approach offers superior performance in terms of infection containment, control efficiency, and disturbance robustness. Extensive numerical results demonstrate the effectiveness of the proposed technique.
\end{abstract}

\noindent 
\textbf{Index Terms}—\textbf{Dissipativity-Based Control, Spreading Processes, Epidemics Control, Spreading Networks Design.}

\section{Introduction} 

% Why epidemic spreading networks?
The study of epidemic spreading processes occurring over large-scale interconnected systems (networks) is a crucial topic in epidemiology \cite{Hethcote2000,Pastor-Satorras2001}, which also influences the planning of resources like social services, healthcare systems, and transportation systems \cite{Barrett2009,Nowzari2016}. 
Moreover, epidemic spreading processes, in terms of their underlying dynamics, share many similarities with the information and opinion spreading processes that occur in computer and social networks, respectively \cite{Sivaraman2023}. 
Clearly, understanding epidemic spreading processes over networks, particularly how infections propagate across multiple groups within a network due to inter-group interactions among nodes in those groups, is vital for designing effective control mechanisms (e.g., reducing inter-group interactions by a certain margin) to prevent undesirable epidemic outbreaks \cite{Wang2024}.
Consequently, such findings on the analysis and control of epidemic spreading processes over networks will have direct consequences across various application domains, including social services, healthcare systems, transportation systems, computer networks, and social networks, particularly for their planning, analysis, and control purposes. 

% Why hierarchical control
The networks over which spreading processes occur are typically modeled as graphs, where nodes represent entities (e.g., individuals, households, local communities) and edges represent interactions that facilitate the spreading process \cite{Hethcote2000,Pastor-Satorras2001}. However, the recent research promotes modeling such networks using hierarchical graphs to efficiently handle their large-scale and multi-level nature \cite{Wang2024b,Tian2022}. Adopting such a hierarchical modeling paradigm enables the hierarchical analysis and control of the epidemic spreading process, e.g., at local, regional, and global levels within the network. However, designing such hierarchical analysis and control strategies remains challenging due to the involved nonlinear node dynamics, unknown disturbances, their complex interconnections, and compositionality and scalability concerns. 

% Conventional methods: Non-pharmaceutical intervention vs pharmaceutical Intervention
Conventional techniques for controlling a spreading process over a network (spreading network) typically focus on modifying the network topology to influence the propagation of infections and/or enhance the robustness of the spreading network \cite{Xue2019,Doostmohammadian2023}. As an alternative to such non-pharmaceutical intervention-based spreading network control methods, some works also consider pharmaceutical intervention-based methods where node characteristics are modified to accelerate local healing (recovery) rates \cite{Gou2017,Preciado2014}. However, such interventions require extensive resources, such as targeted healthcare programs (e.g., vaccinations), as well as preparation time. Therefore, non-pharmaceutical intervention-based methods are more popular and widely studied for controlling spreading networks.

% Static non-pharmaceutical intervention methods
Two commonly used non-pharmaceutical intervention-based spreading network control methods are threshold pruning and high-degree edge removal \cite{Mieghem2011,Schneider2011}. Threshold pruning eliminates inter-group connections that exceed a specified threshold, thereby removing the most impactful pathways for the spread of infection \cite{Enns2015}. High-degree edge removal eliminates outgoing connections from nodes with high inter-group connectivity, thereby limiting the influence of highly connected nodes \cite{Xu2021c}. While these non-pharmaceutical intervention-based spreading network control methods are very efficient and practical, they are typically heuristic and static. Therefore, they fail to adequately account for the complex dynamics and external disturbances present in a real-world spreading network, and they do not provide guarantees regarding the stability or robustness of the spreading network.

% Dynamic feedback control (non-pharmaceutical) methods
In contrast to such static spreading network control methods, dynamic feedback control methods model the spreading network as a dynamical system and design state-dependent control laws that suppress the spread of infection. The work in \cite{Zhang2024} proposes a state feedback controller, along with an observer, to tune the effective infection/recovery dynamics. The work in \cite{Walsh2025} proposes a decentralized adaptive gain controllers that provably drive the spreading network to a disease-free state using only local measurements. 
To reduce communication and computation requirements while retaining stability properties in spreading network control, the work in \cite{Hashimoto2021} has proposed an event-triggered and sparse update control strategy. There also exist spreading network control methods that exploit distributed optimal control \cite{Kovacevic2022} and model predictive control \cite{Carli2020}.

% Drawbacks of Dynamic feedback control-based methods
While the aforementioned dynamic feedback control methods account for the underlying complex dynamics, improve the responsiveness and disturbance robustness, and provide formal stability guarantees under specific conditions, they face several \underline{B}asic \underline{C}hallenges: 
(BC1) Compositionality and scalability challenges in both control design and operation stages when applied to large-scale hierarchical spreading networks; 
(BC2) Lack of rigorous disturbance robustness guarantees while accounting for both non-linear node dynamics and their complex hierarchical interconnection structure.    
(BC3) Lack of flexibility to efficiently co-design the spreading network topology and specific control efforts. 
To address these challenges BC1-BC3, inspired by \cite{SheHale2024}, we propose a dissipativity-based control strategy for spreading networks.

% How dissipativity can address these challenges
Dissipativity theory provides a compositional and scalable energy-based framework for analyzing and controlling generic networked systems \cite{WelikalaP52022}. The underlying idea is to first analyze/enforce each subsystem of the networked system with a storage function (representing its stored energy) - the rate of which is bounded by a quadratic supply rate function (representing its input-output power flow). Then, such functions are composed to analyze/enforce storage and supply rate functions for the networked system. By analyzing/enforcing the validity of the networked system storage function under different supply rate functions, corresponding input-output certificates, such as passivity and $L_2$-stability (which imply disturbance robustness guarantees), can be analyzed/enforced for the networked system. Hence, dissipativity theory also enables the analysis and enforcement of disturbance robustness guarantees for networked systems analyzing and controlling generic networked systems \cite{WelikalaP52022}. Moreover, as shown in \cite{WelikalaP52022}, dissipativity theory also enables efficient co-design of topologies and controllers for generic networked systems. In all, dissipativity theory has proven to be an effective tool for addressing similar challenges to those identified above in BC1-BC3 when analyzing and controlling generic networked systems \cite{WelikalaP52022}.

% How does this work compare to our other work on NCPS 
The theoretical results identified in \cite{WelikalaP52022} have already been significantly specialized and extended to address specific analysis and design problems arising in networked cyber-physical systems such as platoons \cite{WelikalaP72023}, supply chains \cite{Welikala2025Ax1}, and microgrids \cite{Najafirad2025P1}. However, so far, they have not been applied to handle the analysis and control of spreading networks due to several \underline{A}pplication-specific \underline{C}hallenges: 
(AC1) Node dynamics are non-linear and cannot be altered through local controllers (without pharmaceutical interventions); 
(AC2) Multi-level hierarchical nature of the spreading networks;
(AC3) Lack of established disturbance modeling approaches and the need for advanced disturbance robustness guarantees like Scalable Mesh Stability (SMS) \cite{Mirabilio2022}; and
(AC4) Existing hard constraints on the designed spreading network topology (with respect to the nominal topology). 
Therefore, we dedicate this paper to addressing the application-specific challenges, AC1-AC4, in order to develop a dissipativity-based analysis and control technique for spreading networks that can overcome the said basic challenges, BC1-BC3.

Our contributions can be summarized as follows.
\begin{enumerate}[leftmargin=*]
\item We model epidemic spreading processes (i.e., spreading networks) as hierarchical networked systems (comprised of, not limited to, groups and nodes) with disturbance inputs and performance outputs. 
\item We derive dissipativity properties of the nodes subject to disturbances present in the spreading network, which are then subsequently used to analyze the dissipativity of any group of such nodes, and of any network of such groups. 
\item We also propose a spreading network design (i.e., control) technique to enforce dissipativity (which includes stability and disturbance robustness) by optimally modifying the network topology subject to any existing hard constraints with respect to the nominal topology. 
\item We provide how advanced disturbance robustness properties like scalable mesh stability (SMS) \cite{Mirabilio2022} can be analyzed or enforced for both generic networked systems and spreading networked systems.  
\item The proposed analysis and control techniques do not involve solving computationally expensive iterative multi-level optimization problems.  Instead, we strategically include specifically designed necessary conditions at each optimization problem to support the feasibility of the optimization problem solved at the subsequent level. 
\item All optimization problems take the form of convex Linear Matrix Inequality (LMI) problems, and hence they can be solved efficiently using existing software toolboxes \cite{Lofberg2004}. 
\item We provide clear implementation details and extensive numerical results that showcase the benefits of the proposed method and compare it with two basic spreading network control methods.
\end{enumerate}

This paper is structured as follows. Section \ref{Sec:Preliminaries} provides necessary preliminary results on dissipativity and networked systems, and they are further extended in Section \ref{Sec:Preliminaries2}. Section \ref{Sec:Formulation} presents the spreading network dynamics and formulates the analysis and control (design) problems. The proposed dissipativity-based analysis and control techniques are detailed in Section \ref{Sec:Solution}. Simulation results are presented in Section \ref{Sec:Simulation} before concluding the paper in Section \ref{Sec:Conclusion}. 
% All proofs are omitted due to space constraints but can be found in the extended version of this paper \cite{Najafirad2}. 

\section{Preliminaries}
\label{Sec:Preliminaries}

We start by providing several preliminary results that will be extended and utilized in the later sections of this paper. 

\subsection{Notations}
We use $\mathbb{R}$ and $\mathbb{N}$ to denote sets of real and natural numbers, respectively. For any $N\in\mathbb{N}$, we define $\mathbb{N}_N\triangleq\{1,2,..,N\}$.
An $n \times m$ block matrix $A$ is denoted either as $A = [A_{ij}]_{i \in \mathbb{N}_n, j \in \mathbb{N}_m}$ or $A = [A^{ij}]_{i \in \mathbb{N}_n, j \in \mathbb{N}_m}$. We consider $[A_{ij}]_{j\in\mathbb{N}_m}$ and $\diag([A_{ii}]_{i\in\mathbb{N}_n})$ to be a block row matrix and a block diagonal matrix, respectively.
$\0$ and $\I$ respectively denote zero and identity matrices (their dimensions will be clear from the context). A symmetric positive definite (semi-definite) matrix $A\in\mathbb{R}^{n\times n}$ is denoted by $A>0\ (A\geq0)$. We use the symbol $\star$ to represent conjugate matrices in symmetric matrices (their meaning will be clear from the context). We also define $\mathcal{H}(A) \triangleq A + A^\T$ (where $A^\T$ is the transpose of $A$), $\mb{1}_{\{ \cdot \}}$ as the indicator function and $\e_{ij} \triangleq \mb{1}_{\{i=j\}}$. 
The notations $\mathcal{K}_\infty$ and $\mathcal{KL}$ denote respective class-$\mathcal{K}$ functions \cite{Sontag1995}.   
We use $\vert \cdot \vert$ to denote the Euclidean (spectral) norm of a vector (matrix), and $\Vert \cdot \Vert$ to denote any standard norm (note $\vert \cdot \vert = \Vert \cdot \Vert_2$). For a time-varying vector (i.e., a signal) $x(\cdot) \triangleq \{x(t) \in \R^n\}_{t \geq 0}$, $\mathcal{L}_2$ and $\mathcal{L}_{\infty}$ norms are given by $\Vert x(\cdot) \Vert  \triangleq \sqrt{\int_{0}^{\infty} \vert x(t)\vert^2dt}$ and $\Vert x(\cdot) \Vert_{\infty} \triangleq \sup_{t\geq 0} |x(t)|$, respectively. 
We define $\overline{\lambda}(A)$ and $\underline{\lambda}(A)$ respectively as the largest and the smallest eigenvalues of $A$. We use $\leqew$ to denote the element-wise version of the inequality $\leq$. For notational convenience, we will omit explicitly denoting the time dependence of variables in some obvious places.

\subsection{Dissipativity}
Consider the non-linear dynamical system described by
\begin{equation}\label{Eq:GeneralSystem}
\begin{aligned}
\dot{x}(t) = f(x(t),u(t)),\\
y(t) = h(x(t),u(t)),
\end{aligned}
\end{equation}
where $x(t)\in\mathbb{R}^{n_x}$, $u(t)\in\mathbb{R}^{n_u}$ and $y(t)\in\mathbb{R}^{n_y}$ denote the state, input and output, respectively. Functions $f:\mathbb{R}^{n_x} \times \mathbb{R}^{n_u}\rightarrow\mathbb{R}^{n_x}$ and $h:\mathbb{R}^{n_x}\times\mathbb{R}^{n_u}\rightarrow\mathbb{R}^{n_y}$ are continuously differentiable, and satisfy $f(\0,\0) = \0$ and $h(\0,\0)= \0$. 

\begin{definition}\label{Def:Dissipativity}
The system \eqref{Eq:GeneralSystem} is called \emph{dissipative} under a supply rate function $S:\mathbb{R}^{n_u}\times\mathbb{R}^{n_y}\rightarrow\mathbb{R}$ if there exists a continuously differentiable storage function $V:\mathbb{R}^{n_x}\rightarrow\mathbb{R}$ such that  $V(x)>0$ when $x\neq 0$, $V(0)=0$, and 
\begin{center}
    $\dot{V}(x)=\nabla_xV(x)f(x,u)\leq S(u,y)$,
\end{center}
for all trajectories $(x,u)\in\mathbb{R}^n\times\mathbb{R}^q$ of the system.
\end{definition}

Based on the used supply rate function $S(u,y)$, the dissipativity property defined in Def. \ref{Def:Dissipativity} can be further specialized. The $X$-dissipativity \cite{WelikalaP52022} (also known as $(Q,S,R)$-dissipativity \cite{Willems1972a,Kottenstette2014}), defined in the sequel, uses a quadratic supply rate function determined by a coefficient matrix $X=X^\top \triangleq [X^{kl}]_{k,l\in\mathbb{N}_2} \in \mathbb{R}^{(n_u + n_y)\times (n_u + n_y)}$.

\begin{definition}\label{Def:X-Dissipativity}
The system \eqref{Eq:GeneralSystem} is $X$-dissipative if it is dissipative under the quadratic supply rate function:
\begin{equation*}
S(u,y)\triangleq
\begin{bmatrix}
    u \\ y
\end{bmatrix}^\top
\begin{bmatrix}
    X^{11} & X^{12}\\ X^{21} & X^{22}
\end{bmatrix}
\begin{bmatrix}
    u \\ y
\end{bmatrix}.
\end{equation*}
\end{definition}

The notion of $X$-dissipativity, as outlined in the following remark, captures several important input-output properties.

\begin{remark}\label{Rm:X-DissipativityVersions}
If the system \eqref{Eq:GeneralSystem} is $X$-dissipative with:
\begin{enumerate}
\item $X = \scriptsize
\begin{bmatrix}
\0 & \frac{1}{2}\I \\ \frac{1}{2}\I & \0
\end{bmatrix}\normalsize$, 
then it is passive;\\
\item $X = \scriptsize
\begin{bmatrix}
-\nu\I & \frac{1}{2}\I \\ \frac{1}{2}\I & -\rho\I
\end{bmatrix}\normalsize$, 
then it is strictly passive with $\nu$ and $\rho$ respectively being the input and output passivity indices (also denoted as IF-OFP($\nu,\rho$));\\
\item $X = 
\scriptsize\begin{bmatrix}
\gamma^2\I & \0 \\ \0 & -\I
\end{bmatrix}\normalsize$, then it is finite-gain $L_2$-stable with $\gamma$ being the $L_2$ gain (also denoted as $L2G(\gamma)$).
\end{enumerate}
\end{remark}

\subsection{Dissipativity Based Networked System Design}
Consider a networked system $\Sigma: w \rightarrow z$ as shown in Fig. \ref{Fig:NetworkedSystem} comprised of $N$ subsystems $\{\Sigma_i:i\in\N_N\}$ along with a disturbance input $w\in\R^{n_w}$ and a performance output $z\in\R^{n_z}$. Each subsystem $\Sigma_i: u_i\rightarrow y_i,i\in\N_N$ follows the dynamics 
\begin{equation}
\label{Eq:SubsystemDynamics}
\begin{aligned}
\dot{x}_i(t) =&\ f_i(x_i(t),u_i(t)),\\
y_i(t) =&\ g_i(x_i(t),u_i(t)),   
\end{aligned}
\end{equation}
where $x_i(t)\in\R^{n_{xi}}$, $u_i(t)\in\R^{n_{ui}}$ and $y_i(t)\in \R^{n_{yi}}$ denotes the state, input and output, respectively, and $f_i(\0,\0) = \0$ and $h_i(\0,\0) = \0$ (similarly to \eqref{Eq:GeneralSystem}). The subsystems are interconnected according to the relationship
\begin{equation}\label{Eq:InterconnectionMatrix}
\bm{u(t)\\ z(t)} = M \bm{y(t)\\w(t)} \equiv \bm{M_{uy} & M_{uw} \\ M_{zy} & M_{zw}} \bm{y(t)\\w(t)}
\end{equation}
where $u(t) \triangleq [u_i^\T(t)]_{i\in\N_N}^\T \in \R^{n_u}$ (with $n_u \triangleq \sum_{i\in\N_N} n_{ui}$), and $y(t) \triangleq [y_i^\T(t)]_{i\in\N_N}^\T \in \R^{n_y}$ (with $n_y \triangleq \sum_{i\in\N_N} n_{yi}$) are vectorized subsystem inputs and outputs, respectively, and $M \in \R^{(n_u + n_z)\times (n_y + n_w)}$ is the interconnection matrix.

\begin{figure}
    \centering
    \includegraphics[width=0.65\linewidth]{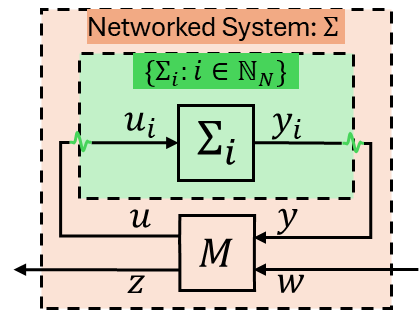}
    \caption{Networked System}
    \label{Fig:NetworkedSystem}
\end{figure}

In the context of this paper, we require the following standard simplifying assumption regarding the dissipativity properties of the subsystems (also used and justified in \cite{WelikalaP52022}). 

\begin{assumption}\label{As:SubsystemDissipativity}
Each subsystem $\Sigma_i,\ i\in\N_N$ \eqref{Eq:SubsystemDynamics} is $X_i$-dissipative (from input $u_i(t)$ to output $y_i(t)$) where $X_i = X_i^\T = [X_i^{kl}]_{k,l\in\N_2}$ is such that $X_i^{11} > 0$. 
\end{assumption}

The following proposition from \cite{WelikalaP52022} exploits subsystem dissipativity properties $\{X_i:i\in\N_N\}$ assumed in As. \ref{As:SubsystemDissipativity} to design the interconnection matrix $M$ so that the networked system is $\X$-dissipative, where $\mathcal{X} = \mathcal{X}^\T \triangleq  [\mathcal{X}^{kl}]_{k,l\in\N_2} \in \R^{(n_w + n_z)\times(n_w + n_z)}$ with $\mathcal{X}^{22}<0$.

\begin{proposition}\cite{WelikalaP52022}\label{Prop:NetworkedSystem}
Under As. \ref{As:SubsystemDissipativity}, the networked system $\Sigma$ shown in Fig. \ref{Fig:NetworkedSystem} can be made $\mathcal{X}$-dissipative (from input $w(t)$ to output $z(t)$, where $\mathcal{X} = \mathcal{X}^\T \triangleq [\mathcal{X}^{kl}]_{k,l\in\N_2}$ with $\mathcal{X}^{22}<0$), by designing the interconnection matrix $M$ via the LMI problem
\begin{equation}\label{Eq:Prop:NetworkedSystem}
    \begin{aligned}
    \mbox{Find: }& L_{uy}, L_{uw}, M_{zy}, M_{zw}, \{p_i: i\in\mathbb{N}_N\}\\
    \mbox{Sub. to: }& p_i > 0,\ \forall i\in\mathbb{N}_N,\ \Psi > 0
    \end{aligned}
\end{equation}
where $\Psi$ is as defined in \eqref{Eq:Prop:NetworkedSystem2},  $\textbf{X}_p^{kl} \triangleq diag([p_iX_i^{kl}]_{i\in\N_N})$ for any $k,l\in\N_2$,  
$\textbf{X}^{12}\triangleq diag([(X_i^{11})^{-1}X_i^{12}]_{i\in\mathbb{N}_N})$, $\textbf{X}^{21}\triangleq (\textbf{X}^{12})^\top$ with 
$M_{uy}\triangleq (\textbf{X}_p^{11})^{-1} L_{uy}$ and $M_{uw} \triangleq   (\textbf{X}_p^{11})^{-1} L_{uw}$.  
\end{proposition}

\begin{figure*}[!hb]
\hrulefill
\centering 
\begin{equation}
\label{Eq:Prop:NetworkedSystem2}
\Psi \triangleq 
\begin{bmatrix} 
\textbf{X}_p^{11} & \textbf{0} & L_{uy} & L_{uw} \\
\star & -\mathcal{X}^{22} & -\mathcal{X}^{22}M_{zy} & -\mathcal{X}^{22} M_{zw}\\ 
\star & \star & -\mathcal{H}(\textbf{X}^{21}L_{uy})-\textbf{X}_p^{22} & -\textbf{X}^{21}L_{uw}+M_{zy}^\top\mathcal{X}^{21} \\
\star & \star & \star &  \mathcal{H} (\mathcal{X}^{12}M_{zw}) + \mathcal{X}^{11}
\end{bmatrix}
\end{equation}
\end{figure*}

If the interconnection matrix $M$ is constrained so that 
\begin{equation}\label{Eq:InterconnectionMatrix2}
M_{uw} \triangleq \I, \quad 
M_{zy} \triangleq \I, \quad 
M_{zw} \triangleq \0,
\end{equation}
(i.e., when $u = M_{uy}y + w$ and $z = y$ in \eqref{Eq:InterconnectionMatrix}), the following corollary can be used to design the remaining interconnection matrix block $M_{uy}$ so that the networked system is $\X$-dissipative.

\begin{corollary}\label{Co:NetworkedSystem}
Under As. \ref{As:SubsystemDissipativity} and constraints \eqref{Eq:InterconnectionMatrix2}, the networked system $\Sigma$ can be made $\mathcal{X}$-dissipative (from input $w(t)$ to output $z(t)$, where $\mathcal{X} = \mathcal{X}^\T =  [\mathcal{X}^{kl}]_{k,l\in\N_2}$ with $\mathcal{X}^{22}<0$), by designing the  interconnection matrix $M$ via the LMI problem
\begin{equation}\label{Eq:Co:NetworkedSystem}
    \begin{aligned}
    \mbox{Find: }& L_{uy}, \{p_i: i\in\mathbb{N}_N\}, \X\\
    \mbox{Sub. to: }& p_i > 0,\ \forall i\in\mathbb{N}_N,\ \Phi > 0
    \end{aligned}
\end{equation}
with $M_{uy}\triangleq (\textbf{X}_p^{11})^{-1} L_{uy}$, where $\Phi$ takes the form
\begin{equation}
\label{Eq:Co:NetworkedSystem2}
\Phi \triangleq 
\begin{bmatrix} 
\textbf{X}_p^{11} & \textbf{0} & L_{uy} & \textbf{X}_p^{11} \\
\star & -\mathcal{X}^{22} & -\mathcal{X}^{22} & \0\\ 
\star & \star & -\mathcal{H}(\textbf{X}^{21}L_{uy})-\textbf{X}_p^{22} & -\textbf{X}_p^{21}+\mathcal{X}^{21} \\
\star & \star & \star &  \mathcal{X}^{11}
\end{bmatrix}.
\end{equation}
\end{corollary}
\begin{proof}
As $M_{uw}\triangleq (\textbf{X}_p^{11})^{-1} L_{uw}$ and $M_{uw} \triangleq \I$, we get $L_{uw} = \textbf{X}_p^{11}$. Note also that $\textbf{X}^{12} \triangleq (\textbf{X}_p^{11})^{-1}\textbf{X}_p^{12}$, and thus, $\textbf{X}_p^{11}\textbf{X}^{12}=\textbf{X}_p^{12}$. Similarly,   $\textbf{X}^{21}\textbf{X}_p^{11}=\textbf{X}_p^{21}$. Applying these results together with \eqref{Eq:InterconnectionMatrix2} in \eqref{Eq:Prop:NetworkedSystem}, we can obtain \eqref{Eq:Co:NetworkedSystem}.
\end{proof}

\begin{remark}\label{Rm:Analysis}
For $\X$-dissipativity analysis of the networked system using Prop. \ref{Prop:NetworkedSystem}, one can follow a similar process to the Proof of Co. \ref{Co:NetworkedSystem}. It simply involves setting $M$ to its nominal value and then executing some minor change of variables, before evaluating the feasibility of the transformed version of the LMI problem  \eqref{Eq:Prop:NetworkedSystem}.  
\end{remark}

\begin{remark}
In Co. \ref{Co:NetworkedSystem}, the parameter $\X$ in the networked system $\X$-dissipativity is considered as a decision variable since its components (i.e., $\{\X^{kl}:k,l\in\N_2\}$) appear linearly in \eqref{Eq:Co:NetworkedSystem}. This introduces an extra degree of freedom enabling the optimization of the enforced (or analyzed) $\X$-dissipativity properties. This fact is particularly useful when the interested $\X$ dissipativity property represents input-output properties like IF-OFP($\nu,\rho$) or L2G($\gamma$) outlined in Rm. \ref{Rm:X-DissipativityVersions}.        
\end{remark}

\section{Networked System Design Considerations}
\label{Sec:Preliminaries2}

In this section, we extend the preliminary results presented in Section \ref{Sec:Preliminaries}, particularly the $\X$-dissipative networked system design techniques given in Co. \ref{Co:NetworkedSystem}, to address several design considerations that will be useful in the sequel. 

\subsection{Necessary Conditions to be Enforced at the Subsystems}

Note that the feasibility of the network-level (i.e., global) LMI problem given in Co. \ref{Co:NetworkedSystem} depends on the assumed subsystem-level (i.e., local) dissipativity properties $\{X_i:i\in\N_N\}$. Often, such local dissipativity properties result from a separate set of corresponding local LMI problems. Therefore, it is worthwhile to identify any necessary conditions that exist for local dissipativity properties to support the feasibility of the global LMI problem, as they can then be integrated into such local LMI problems. In the following proposition, we identify a set of such necessary conditions.

\begin{proposition}\label{Prop:NecessaryConditions}
For the feasibility of the global LMI problem \eqref{Eq:Co:NetworkedSystem} given in Co. \ref{Co:NetworkedSystem}, it is necessary that each subsystem $\Sigma_i, i\in\N_N$ to be $X_i$-dissipative such that $X_i$ satisfies the LMI problem:
\begin{equation}\label{Eq:Prop:NecessaryConditions}
\begin{aligned}
\mbox{Find: }&\ X_i, \{\bar{\mathcal{X}}_{ii}^{kl}:k,l\in\N_2\}, \\
\mbox{Sub. to: }&\ \tilde{\Phi}_{ii} > 0
\end{aligned}
\end{equation}
where $\tilde{\Phi}_{ii}$ takes the form 
\begin{equation}
\label{Eq:Prop:NecessaryConditions2}
\begin{aligned}
&\tilde{\Phi}_{ii} \triangleq \\
&\begin{bmatrix} 
X_i^{11} & \textbf{0} & X_i^{11} M_{uy}^{ii} & X_i^{11} \\
\star & -\bar{\mathcal{X}}^{22}_{ii} & -\bar{\mathcal{X}}^{22}_{ii} & \0\\ 
\star & \star & -\mathcal{H}(X_i^{21}M_{uy}^{ii}) - X_i^{22} & -X_i^{21}+\bar{\mathcal{X}}^{21}_{ii} \\
\star & \star & \star & \bar{\mathcal{X}}^{11}_{ii}
\end{bmatrix},
\end{aligned}
\end{equation}
each $\bar{\mathcal{X}}_{ii}^{kl}$ is a free variable matrix with the format of the $i$\tsup{th} diagonal block of $\mathcal{X}^{kl}$ in \eqref{Eq:Co:NetworkedSystem}, and $M_{uy}^{ii}$ is the $i$-th diagonal block of $M_{uy}$ (i.e., the self-connection matrix of $\Sigma_i$).
\end{proposition}
\begin{proof}
For the block matrix $\Phi = [\Phi_{kl}]_{k,l\in\N_4}$ given in \eqref{Eq:Co:NetworkedSystem2}, the equivalence 
$$\Phi > 0 \iff \bar{\Phi} \triangleq [[\Phi_{kl}^{ij}]_{k,l\in\N_4}]_{i,j\in\N_N} > 0$$
holds, where $\bar{\Phi}$ is the ``block element-wise'' permutation \cite{WelikalaJ22022} of $\Phi$. Using the diagonal blocks of $\bar{\Phi}$, we can identify a set of necessary conditions for $\Phi > 0$ as 
$$
\Phi > 0 \iff \bar{\Phi} > 0 \implies 
\{ \bar{\Phi}_{ii} > 0, \forall i\in\N_N\},
$$
where each $\bar{\Phi}_{ii} \triangleq [\Phi_{kl}^{ii}]_{k,l\in\N_4}$ takes the form in \eqref{Eq:Prop:NecessaryConditionsStep1}. Note that for any $i\in\N_N$, the implication,   
$$M_{uy} = (\textbf{X}_p^{11})^{-1} L_{uy} \implies M_{uy}^{ii} = p_i^{-1}(X_i^{11})^{-1}L_{uy}^{ii}$$
has also been used to replace $L_{uy}^{ii}$ terms in deriving \eqref{Eq:Prop:NecessaryConditionsStep1}. Finally, we use the equivalence 
$$
\bar{\Phi}_{ii} > 0 \iff \tilde{\Phi}_{ii} \triangleq \frac{1}{p_i}\bar{\Phi}_{ii} > 0, 
$$ 
and the notation $\bar{\mathcal{X}}^{kl}_{ii} \triangleq \frac{1}{p_i}\mathcal{X}^{kl}_{ii}, \forall k,l\in\N_2$
to obtain
$$
\Phi > 0 \implies \{\tilde{\Phi}_{ii} \mbox{ as in \eqref{Eq:Prop:NecessaryConditions2}}: \tilde{\Phi}_{ii} > 0,  i\in\N_N \},
$$
i.e., the necessary conditions for $\Phi > 0$. 
\end{proof}

\begin{figure*}[!hb]
\centering 
\hrulefill
\begin{equation}
\label{Eq:Prop:NecessaryConditionsStep1}
\bar{\Phi}_{ii} \triangleq 
\begin{bmatrix} 
p_iX_i^{11} & \textbf{0} & p_i X_i^{11} M_{uy}^{ii} & p_iX_i^{11} \\
\star &  -\mathcal{X}^{22}_{ii} & -\mathcal{X}^{22}_{ii} & \0\\ 
\star & \star & - \mathcal{H}(p_iX_i^{21}M_{uy}^{ii})-p_iX_i^{22} & -p_iX_i^{21}+\mathcal{X}^{21}_{ii} \\
\star & \star & \star &  \mathcal{X}^{11}_{ii}
\end{bmatrix}
\end{equation}
\end{figure*}

\begin{remark}
To include the identified necessary condition  \eqref{Eq:Prop:NecessaryConditions} in the local LMI problem solved at $\Sigma_i, i\in\N_N$ (to analyze/enforce its $X_i$-dissipativity), it is required to know the self-connection matrix $M_{uy}^{ii}$. As we will see in the sequel, this information is known and fixed at the ``subsystems'' considered in the spreading networks analysis and control application studied in this paper.
\end{remark}

\begin{remark}
A further simplified necessary condition than \eqref{Eq:Prop:NecessaryConditions} can be obtained if the matrices $X_i^{kl}$ and $\bar{\mathcal{X}}_{ii}^{kl}$ (i.e., subsystem and network level dissipativity measures) are limited to be scalar matrices, over all $k,l\in\N_2$ and $i\in\N_N$. 
The idea is to reapply the diagonal block extraction and ``block element-wise'' formulation technique used in Prop. \ref{Prop:NecessaryConditions}, for \eqref{Eq:Prop:NecessaryConditions}. 
Note that this simplification will be further significant if $M_{uy}^{ii}$ is a zero matrix or has a zero in its diagonal.
\end{remark}

\subsection{Enforcing Constraints on the Interconnection Matrix $M$}

When designing the interconnection matrix $M$ for the networked system using Prop. \ref{Prop:NetworkedSystem} or Co. \ref{Co:NetworkedSystem}, instead of designing it from scratch, it is often required to minimize or constrain deviations from a known nominal interconnection matrix $\bar{M}$. 
However, including such an objective or constraint in the global LMI problems given in Prop. \ref{Prop:NetworkedSystem} and Co. \ref{Co:NetworkedSystem} make them bilinear matrix inequality (BMI) problems. This is due to the involved change of variables (see below \eqref{Eq:Prop:NetworkedSystem} and \eqref{Eq:Co:NetworkedSystem}), which, in general, take the form 
$$L = X M \iff M = X^{-1}L,$$
where $X$ and $L$ are the actual design variables that appear linearly in the global LMI problems. However, as detailed below, there exist some approximate strategies that we can adopt to minimize an objective function of the form 
\begin{equation}
  \label{Eq:ObjectiveFunctionForm}
  J(M) \triangleq \Vert M - \bar{M} \Vert,
\end{equation}
and enforce constraints of the form (with $\delta_2,\delta_l,\delta_u \in [0,1]$)
\begin{align}
  \label{Eq:ConstraintForm1}
  \Vert M - \bar{M} \Vert_2  \leq \delta_2 \Vert \bar{M} \Vert_2, \mbox{ or }\\
  \label{Eq:ConstraintForm2}
  - \delta_l \bar{M} \leqew M - \bar{M} \leqew  \delta_u \bar{M},
\end{align}
when designing the interconnection matrix $M$ without compromising the convex LMI format of the overall problem.

\begin{proposition}\label{Pr:AlternativeObjective}
The objective function $J(M)$ in \eqref{Eq:ObjectiveFunctionForm} can be minimized via minimizing the convex objective function 
\begin{equation}\label{Eq:Pr:AlternativeObjective}
\hat{J}(L,X) \triangleq  \Vert L - X\bar{M} \Vert - \lambda \Vert X \Vert,
\end{equation}
where $\lambda > 0$ if $X$ is a scalar matrix or $\lambda = 0$ otherwise. 
\end{proposition}
\begin{proof}
Using the fact $M = X^{-1}L$, from \eqref{Eq:ObjectiveFunctionForm}, we get 
$$
J(M) = \Vert M - \bar{M}\Vert = \Vert X^{-1}L - X^{-1}X \bar{M}\Vert \Vert X \Vert \Vert X\Vert^{-1}.
$$
Using the identity: $\Vert AB \Vert \leq \Vert A \Vert \Vert B \Vert$, for any conformable $A,B$ matrices, and its consequence $\Vert X \Vert^{-1}  \leq \Vert X^{-1} \Vert$, we can identify lower and upper bounds for $J(M)$ as  
$$  
\Vert X \Vert^{-1} \Vert L - X\bar{M} \Vert \leq  J(M) \leq \Vert X^{-1} \Vert \Vert L - X\bar{M} \Vert.
$$
From these bounds, it is clear that $J(M)$ can be minimized via: (i) minimizing $\Vert L - X\bar{M} \Vert$, which is convex in $L$ and $X$, and (ii) maximizing $\Vert X \Vert$, which is linear (convex and concave) in $X$ when $X$ is a scalar matrix. Combining these two facts, we can obtain the alternative convex objective function $\hat{J}(L,X)$ given in \eqref{Eq:Pr:AlternativeObjective}.  
\end{proof}

\begin{proposition}\label{Pr:AlternativeConstraints}
The constraint \eqref{Eq:ConstraintForm1} on $M$ holds if there exists feasible $L$ and $X$ such that
\begin{equation}\label{Eq:Pr:AlternativeConstraints1}
0 \leq 
\bm{
\I & Y^\T & \0\\
Y & XY^\T + YX^\T  & L-X\bar{M} \\ 
\0 & L^\T - \bar{M}^\T X^\T & \delta^2_2 \bar{M}^\T \bar{M}},
\end{equation}
for some $Y$ (e.g., $Y = \I$). The constraint  \eqref{Eq:ConstraintForm1} on $M$ holds if and only if there exists feasible $L$ and $X$ such that 
\begin{equation}\label{Eq:Pr:AlternativeConstraints2}
- \delta_l X\bar{M} \leqew L - X\bar{M} \leqew  \delta_u X\bar{M},
\end{equation}
when $X$ is diagonal and $X > 0$.
\end{proposition}
\begin{proof}
To prove the first part, we require establishing \eqref{Eq:ConstraintForm1} $\impliedby$ \eqref{Eq:Pr:AlternativeConstraints1}. To this end, note that 
\begin{align*}
\eqref{Eq:ConstraintForm1} 
&\ \iff 
0 \leq  \delta^2_2 \bar{M}^\T \bar{M} - (M-\bar{M})^\T (M-\bar{M}) \\
&\ \iff  
0 \leq \bm{\I & M-\bar{M} \\ M^\T-\bar{M}^\T & \delta^2_2 \bar{M}^\T \bar{M}}\\
&\ \iff  
0 \leq \bm{XX^\T & L-X\bar{M} \\ L^\T - \bar{M}^\T X^\T & \delta^2_2 \bar{M}^\T \bar{M}},
\end{align*}
where the three steps respectively follow from a well-known matrix 2-norm characteristic, the Schur complement, and the congruence principle. From the Cauchy-Schwarz inequality, we have  
$
XY^\T + YX^\T - YY^\T \leq XX^\T.
$
Applying this and reusing the Schur complement, we get
\begin{align*}
0 \leq &\bm{XX^\T & L-X\bar{M} \\ L^\T - \bar{M}^\T X^\T & \delta^2_2 \bar{M}^\T \bar{M}}\\
&\ \impliedby
0 \leq \bm{XY^\T + YX^\T - YY^\T & L-X\bar{M} \\ L^\T - \bar{M}^\T X^\T & \delta^2_2 \bar{M}^\T \bar{M}}\\
&\ \iff \eqref{Eq:Pr:AlternativeConstraints1}.
\end{align*}
Combining these two results, we get \eqref{Eq:ConstraintForm1} $\impliedby$ \eqref{Eq:Pr:AlternativeConstraints1}.

To prove the second part, we require establishing \eqref{Eq:ConstraintForm2} $\iff$ \eqref{Eq:Pr:AlternativeConstraints2}. To this end, we use the given facts that $X$ is diagonal and $X>0$ to multiply \eqref{Eq:ConstraintForm2} by $X$, which gives
$$\eqref{Eq:ConstraintForm2} \iff - \delta_l X\bar{M} \leqew XM - X\bar{M} \leqew  \delta_u X\bar{M} \iff \eqref{Eq:Pr:AlternativeConstraints2},$$
where the last step is from applying $L = XM$. 
\end{proof}

\begin{remark}
As we will see in the sequel, the additional conditions imposed on $X$ in Props. \ref{Pr:AlternativeObjective} and \ref{Pr:AlternativeConstraints} typically hold when Prop. \ref{Prop:NetworkedSystem} or Co. \ref{Co:NetworkedSystem} is applied to designing the interconnection matrix $M$ in spreading network control applications.
\end{remark}

\begin{remark}\label{Rm:BMIforY}
The identified alternative constraint \eqref{Eq:Pr:AlternativeConstraints1} is equivalent to the desired constraint \eqref{Eq:ConstraintForm1} if $Y$ can be selected such that $Y=X$. 
However, such a selection is not possible because only $X$ and $L$ are the design variables in \eqref{Eq:Pr:AlternativeConstraints1} (otherwise \eqref{Eq:Pr:AlternativeConstraints1} would be a BMI in $X,L$ and $Y$). Nevertheless, as an alternative, an efficient iterative process can be created where in each iteration, a candidate solution for the design variable $X$ is first determined and then \eqref{Eq:Pr:AlternativeConstraints1} is re-solved to update $Y$ (note that \eqref{Eq:Pr:AlternativeConstraints1} is linear in $Y$ for a fixed $X$), so that $Y\simeq X$ is achieved at convergence. For more details on such BMI solution techniques based on iterative LMI solutions, we refer the readers to \cite{VanAntwerp2000,Simon2011,Lee2019}. 
\end{remark}

\subsection{Scalable Mesh Stability}

Enforcing dissipativity in a networked system ensures the dissipation of disturbances. In contrast, enforcing scalable mesh stability (SMS), as defined below, ensures the dispersion of the disturbances over the network. Hence, enforcing both dissipativity and SMS can significantly improve the robustness and resilience of a networked system. Motivated by this, we identify necessary conditions for SMS that can be directly incorporated into the dissipative interconnection matrix $M$ design method given in Co. \ref{Co:NetworkedSystem}. 

First, let us consider the networked system $\Sigma$ shown in Fig. \ref{Fig:NetworkedSystem} under the constraints \eqref{Eq:InterconnectionMatrix2} (i.e., the networked system considered in Co. \ref{Co:NetworkedSystem}). In this setting, each subsystem $\Sigma_i,i\in\N_N$ \eqref{Eq:SubsystemDynamics} can be described by (with a slight abuse of notation)
\begin{equation}\label{Eq:SubsystemDynamics2}
\dot{x}_i(t) = f_i(x_i(t),u_i(t)) = f_i(x_i(t),\{x_j(t):j\in\mathcal{N}_i\},w_i(t)),
\end{equation}
were $u_i(t) = \sum_{j\in\N_N}M_{uy}^{ij} x_j(t) + w_i(t)$ and $\mathcal{N}_i\triangleq \{j\in\N_N: M_{uy}^{ij} \neq \0, j\neq i\}$ denotes the set of neighbors of each subsystem $\Sigma_i, i\in\N_N$. Next, we formally recall the definitions of input-to-state stability (ISS) and SMS, along with the necessary condition derived for SMS given in \cite{Mirabilio2022}.

\begin{definition}\label{Def:ISS} \cite{Mirabilio2022}
A subsystem $\Sigma_i, i\in\N_N$ \eqref{Eq:SubsystemDynamics2} is ISS if there exists $\beta_i\in\mathcal{KL},\ \gamma_i \in \mathcal{K}_{\infty}$ and $\sigma_i\in\mathcal{K}_{\infty}$ such that  
$$
\vert x_i(t) \vert \leq \beta_i(\vert x_i(t_0)\vert, t-t_0) + \gamma_i(\max_{j\in\mathcal{N}_i} \Vert x_i \Vert_\infty) + \sigma_i(\Vert w_i \Vert_{\infty}).
$$
\end{definition}

\begin{definition}\label{Def:SMS} \cite{Mirabilio2022}
The networked system comprised of the subsystems \eqref{Eq:SubsystemDynamics2} is said to be scalable mesh stable (SMS) if the ISS property of each subsystem implies the existence of some $\beta \in \mathcal{KL}$ and $\sigma\in\mathcal{K}_{\infty}$ such that
$$
\max_{i\in\N_N} \vert x_i(t) \vert \leq \beta(\max_{i\in\N_N} \vert x_i(t_0)\vert,t-t_0) + \sigma(\max_{i\in\N_N} \Vert w_i\Vert_{\infty}),
$$
for any initial condition $x_i(t_0)$ and disturbance $w_i, i\in\N_N$. 
\end{definition}

\begin{proposition} \cite{Mirabilio2022}
\label{Pr:SMS}
The networked system comprised of the subsystems \eqref{Eq:SubsystemDynamics2} is SMS (see Def. \ref{Def:SMS}) if each subsystem is ISS such that there exists $\tilde{\gamma} \in (0,1)$ that satisfies  
$$
\gamma_i(s) \leq \tilde{\gamma}s,\quad \forall s \geq 0, i\in\N_N. 
$$
\end{proposition}

%% Property of X
% We recall that $X$ is block diagonal (even diagonal) and satisfies $X = X^\T > 0$. 

%% Schur Complement result
% $$
% A > 0,\ C - B^\T A B > 0 \iff \bm{A^{-1} & B \\ B^\T & C} > 0
% $$

Using Prop. \ref{Pr:SMS}, we next derive the additional conditions that should be included in the dissipative networked system design problem given in Co. \ref{Co:NetworkedSystem}. However, we first require the following slightly stronger assumption (than As. \ref{As:SubsystemDissipativity}) regarding the nature of the subsystem dissipativity.

\begin{assumption}\label{As:SubsystemDissipativity2}
Each subsystem $\Sigma_i, i\in\N_N$ \eqref{Eq:SubsystemDynamics2} is $X_i$-dissipative (from input $u_i(t)$ to output $x_i(t)$) where $X_i = X_i^\T = [X_i^{kl}]_{k,l\in\N_2}$ is such that $X_i^{11} > 0$, and the corresponding storage function is $V_i(x_i) \triangleq x_i^\T P_i x_i$ where $P_i > 0$. 
\end{assumption}

%% Notice that we need to use the below result not only at the network feel to enforce SMS. It can also be applied at the group level to analyze SMS of the group. SO lets keep things general here.

\begin{proposition}\label{Pr:NetworkedSystemSMS}
The networked system comprised of the subsystems \eqref{Eq:SubsystemDynamics2}, under As. \ref{As:SubsystemDissipativity2} and constraints \eqref{Eq:InterconnectionMatrix2}, can be made both $\mathcal{X}$-dissipative (from input $w(t)$ to output $z(t)$, where $\mathcal{X} = \mathcal{X}^\T =  [\mathcal{X}^{kl}]_{k,l\in\N_2}$ with $\mathcal{X}^{22}<0$) and SMS, by designing the interconnection matrix $M$ via the LMI problem
\begin{equation}\label{Eq:Pr:NetworkedSystemSMS}
    \begin{aligned}
    \mbox{Find: }& L_{uy}, \{p_i: i\in\mathbb{N}_N\}, \X\\
    \mbox{Sub. to: }& p_i > 0,\ \forall i\in\mathbb{N}_N,\ \Phi > 0, \\
    & \tilde{\lambda}_{i1}
    \sum_{j\in\N_N} \vert (X_i^{11})^{-1} L_{uy}^{ij} \vert < p_i,\ \forall i \in \N_N
    \end{aligned}
\end{equation}
where $\Phi$ is as defined in \eqref{Eq:Co:NetworkedSystem2}, 
$\tilde{\lambda}_{i1} \triangleq \sqrt{\frac{\overline{\lambda}(R_i)\overline{\lambda}(P_i)}{\underline{\lambda}(P_i)\underline{\lambda}(Q_i)}}$ (where $P_i$ is from As. \ref{As:SubsystemDissipativity2}, and $Q_i$ and $R_i$ satisfy $(X_i^{22}+X_i^{21}) \leq -Q_i < 0$ and $(X_i^{11} + X_i^{12}) \leq R_i$), and 
$M_{uy}\triangleq (\textbf{X}_p^{11})^{-1} L_{uy}$
\end{proposition}
\begin{proof}
Compared to Co. \ref{Co:NetworkedSystem}, here we only have to establish the SMS property for the networked system.  

According to As. \ref{As:SubsystemDissipativity2}, as each subsystem $\Sigma_i, i\in\N_N$ is $X_i$-dissipative, we have 
\begin{align}
\dot{V}_i(x_i) 
&\leq x_i^\T X_i^{22} x_i + x_i^\T X_i^{21} u_i + x_i^\T X_i^{21} u_i  + u_i^\T X_i^{11} u_i \nonumber \\
&\leq x_i (X_i^{22}+X_i^{21}) x_i + u_i^\T (X_i^{11} + X_i^{12}) u_i \nonumber \\
&\leq - x_i Q_i x_i + u_i^\T R_i u_i \nonumber \\
&\leq - \underline{\lambda}(Q_i) \vert x_i \vert^2 + \overline{\lambda}(R_i) \vert u_i \vert^2, \label{Eq:Pr:NetworkedSystemSMSStep1}
\end{align}
where we have selected $Q_i$ and $R_i$ so that
\begin{equation}\label{Eq:Pr:NetworkedSystemSMSStep2}
  (X_i^{22}+X_i^{21}) \leq -Q_i < 0\  \mbox{ and }\ 
  (X_i^{11} + X_i^{12}) \leq R_i.
\end{equation}
Moreover, as each subsystem $\Sigma_i, i\in\N_N$ has a quadratic storage function of the form $V_i(x_i) = x_i^\T P_i x_i$, we have
\begin{equation}\label{Eq:Pr:NetworkedSystemSMSStep3}
\underline{\lambda}(P_i) \vert x_i \vert^2 
\leq V_i(x_i) \leq 
\overline{\lambda}(P_i) \vert x_i \vert^2. 
\end{equation}
Applying \eqref{Eq:Pr:NetworkedSystemSMSStep3} in \eqref{Eq:Pr:NetworkedSystemSMSStep1} leads to 
\begin{align}
\dot{V}_i(x_i) 
&\leq - \frac{\underline{\lambda}(Q_i)}{\overline{\lambda}(P_i)} V_i(x_i) + 
\overline{\lambda}(R_i) \vert u_i \vert^2 \nonumber \\
&= -\mu_i V_i(x_i) + \theta_i, 
\label{Eq:Pr:NetworkedSystemSMSStep4}
\end{align}
where we define 
$\mu_i \triangleq \frac{\underline{\lambda}(Q_i)}{\overline{\lambda}(P_i)}$ and $\theta_i \triangleq \overline{\lambda}(R_i) \vert u_i \vert^2$. Note that \eqref{Eq:Pr:NetworkedSystemSMSStep4} implies that
\begin{align} 
V_i(x_i) 
&\leq (1-e^{-\mu_i t}) \frac{\theta_i}{\mu_i} + V_i(x_i(0))e^{-\mu_i t} \nonumber \\
&\leq \frac{\theta_i}{\mu_i} 
+ \overline{\lambda}(P_i) \vert x_i(0) \vert^2 e^{-\mu_i t}, \label{Eq:Pr:NetworkedSystemSMSStep5}
\end{align}
where we have further used the fact $(1-e^{-\mu_i t}) \leq 1$ and \eqref{Eq:Pr:NetworkedSystemSMSStep3}. Using \eqref{Eq:Pr:NetworkedSystemSMSStep3} and the definitions of $\theta_i$ and $\mu_i$, \eqref{Eq:Pr:NetworkedSystemSMSStep5} implies
\begin{equation}\nonumber
\vert x_i(t) \vert^2 \leq \frac{\overline{\lambda}(R_i) \overline{\lambda}(P_i)}{\underline{\lambda}(P_i) \underline{\lambda}(Q_i)}\vert u_i \vert^2 
+ \frac{\overline{\lambda}(P_i)}{\underline{\lambda}(P_i)} \vert x_i(0) \vert^2 e^{-\mu_i t},
\end{equation}
which, using the fact that $\sqrt{a+b} \leq \sqrt{a} + \sqrt{b}$, leads to
\begin{equation}\label{Eq:Pr:NetworkedSystemSMSStep6}
\vert x_i(t) \vert \leq \sqrt{\frac{\overline{\lambda}(R_i)\overline{\lambda}(P_i)}{\underline{\lambda}(P_i)\underline{\lambda}(Q_i)}} \vert u_i \vert
+ \sqrt{\frac{\overline{\lambda}(P_i)}{\underline{\lambda}(P_i)}} \vert x_i(0) \vert e^{-\frac{\mu_i}{2} t}.
\end{equation}

We will next bound the term $\vert u_i \vert$. Using the relationship $u_i = \sum_{j\in\N_N}M_{uy}^{ij} x_j + w_i$, we get 
\begin{align}
\vert u_i \vert 
&\leq \sum_{j\in\N_N} \vert M_{uy}^{ij} x_j \vert + \vert w_i \vert, \nonumber \\
&\leq \sum_{j\in\N_N} \vert M_{uy}^{ij} \vert \vert x_j \vert + \vert w_i \vert \nonumber \\
&\leq \sum_{j\in\N_N} \vert M_{uy}^{ij} \vert \max_{j\in\N_N} \Vert x_j \Vert_{\infty} + \Vert w_i \Vert_{\infty}.\nonumber
\end{align}
Applying this result in \eqref{Eq:Pr:NetworkedSystemSMSStep6}, we obtain 
\begin{equation}\label{Eq:Pr:NetworkedSystemSMSStep7}
\begin{aligned}
\vert x_i(t) \vert \leq& 
\tilde{\lambda}_{i1}
\sum_{j\in\N_N} \vert M_{uy}^{ij} \vert \max_{j\in\N_N} \Vert x_j \Vert_{\infty} + 
\tilde{\lambda}_{i1} \Vert w_i \Vert_{\infty}\\
& + \tilde{\lambda}_{i2} \vert x_i(0) \vert e^{-\frac{\mu_i}{2} t},
\end{aligned}
\end{equation}
where 
$$
\tilde{\lambda}_{i1} \triangleq \sqrt{\frac{\overline{\lambda}(R_i)\overline{\lambda}(P_i)}{\underline{\lambda}(P_i)\underline{\lambda}(Q_i)}}  
\mbox{ and }
\tilde{\lambda}_{i2} \triangleq \sqrt{\frac{\overline{\lambda}(P_i)}{\underline{\lambda}(P_i)}}.
$$

Based on Prop. \ref{Pr:SMS} and Def. \ref{Def:ISS}, the networked system to be SMS, we require the existence of functions $\beta_i\in\mathcal{KL}$ and $\sigma_i\in\mathcal{K}_{\infty}$ such that  
$$
\vert x_i(t) \vert \leq \beta_i(\vert x_i(t_0)\vert, t-t_0) + \tilde{\gamma} \max_{j\in\N_N} \Vert x_i \Vert_\infty + \sigma_i(\Vert w_i \Vert_{\infty}),
$$
for all $i\in\N_N$ where $\tilde{\gamma}$ is a scalar $\tilde{\gamma} \in (0,1)$. Comparing this with \eqref{Eq:Pr:NetworkedSystemSMSStep7} we can obtain the condition for the SMS of the networked system as  
\begin{equation*}
  \tilde{\lambda}_{i1}
  \sum_{j\in\N_N} \vert M_{uy}^{ij} \vert < 1, \quad \forall i\in\N_N.
\end{equation*}
Note that, under the change of variables used in Co. \ref{Co:NetworkedSystem},
$M_{uy}\triangleq (\textbf{X}_p^{11})^{-1} L_{uy} \iff M_{uy}^{ij} = (p_iX_i^{11})^{-1} L_{uy}^{ij}, \ \forall i,j\in\N_N.$
Applying this in the previously identified condition for the SMS of the networked system and multiplying both sides by $p_i$, we can obtain the last constraint in \eqref{Eq:Pr:NetworkedSystemSMS} - which is the new constraint included for SMS, compared to Co. \ref{Co:NetworkedSystem}.
\end{proof}

\subsection{Dissipativity of a Generic Node}

As we will see in Sec. \ref{Sec:Formulation}, a spreading network is consists of an interconnected set of nodes. Before moving on to the next sections where the spreading network analysis/control problem formulation and its solution will be presented (in Sections \ref{Sec:Formulation} and \ref{Sec:Solution}, respectively), here we require establishing a final result regarding nodes. In particular, we consider a generic node and analyze its dissipativity properties under some bounded uncertainties - to obtain a result that will be useful in the subsequent Sec. \ref{Sec:Solution}.

Consider a scalar dynamic system $\Sigma_s:u(t) \rightarrow y(t)$ described by (for known and fixed $\bar{\gamma}, \delta > 0$):
\begin{equation}\label{Eq:ScalarSystem}
\begin{aligned}
    \dot{x}(t) =&\ -\gamma(t) x(t) + (1-x(t))u(t)\\
    y(t) =&\ x(t)
\end{aligned}
\end{equation}
where $\vert \gamma(t)-\bar{\gamma}\vert \leq \delta$, $x(t)\in[0,1]$ and $u(t)\in\R_{\geq0}$. 

As we will see in Sec. \ref{Sec:Formulation}, the scalar dynamic system $\Sigma_s$ \eqref{Eq:ScalarSystem} follows identical dynamics to that of a generic node in a spreading network considered in this paper. The following proposition analyzes the $X$-dissipativity of $\Sigma_s$ \eqref{Eq:ScalarSystem}.

%%%% Q: Can we apply s-lemma here? not that straightforward due ot the involved nonlinear terms

\begin{proposition}\label{Prop:ScalarSystem}
The scalar dynamic system $\Sigma_s$ \eqref{Eq:ScalarSystem}  is $X$-dissipative with any $X=X^\T \triangleq \scriptsize \bm{a & b\\ b & c}\in\R^{2\times 2}$ that satisfies the LMI problem:
\begin{equation}
\label{Eq:Prop:ScalarSystem}
\begin{aligned}
&\mbox{Find: }\ p,\ a,\ b,\ c,\  \\ &\mbox{Sub. to: }\ p > 0,\\
&\begin{cases}
a \geq 0,\ (2b-p)\geq 0,\ (c + p(\bar{\gamma}-\delta)) \geq 0,\ \mbox{or}\\
a > 0,\ b\leq 0,\ \bm{a & b\\ b & c + p(\bar{\gamma}-\delta)} \geq 0.
\end{cases}
\end{aligned}
\end{equation}
\end{proposition}
\begin{proof}
Note that $x=0$ is an equilibrium point of $\Sigma_s$ (with $u=0$ and $y=0$). Therefore, as required in Definitions \ref{Def:Dissipativity} and \ref{Def:X-Dissipativity}, consider a quadratic storage function $V(x)$ and a quadratic supply rate function $S(u,y)$ defined as  
$$ 
V(x) \triangleq \frac{1}{2} p x^2 \quad \mbox{ and } \quad
S(u,y) \triangleq \bm{u \\ y}^\T  \bm{a & b\\ b & c} \bm{u \\ y},
$$
respectively, where $p, a, b, c \in \R$ are design variables and $p > 0$. And for $X$-dissipativity of $\Sigma_s$, we require
$$\dot{V}(x) \leq  s(u,y) \iff 0 \leq \Phi(u,x) \triangleq s(u,x) - \dot{V}(x)$$ 
to hold for all possible trajectories of the system. Using the fact $\dot{V}(x) = px\dot{x}$ with \eqref{Eq:ScalarSystem}, we can simplify $\Phi(u,x)$ as 
\begin{align*}
\Phi(u,x)
% =&\ (au^2 + 2bux + cx^2) - px(-\gamma x + (1-x)u) \\
=&\ au^2 + ((2b-p) + px)xu + (c + p\gamma)x^2.    
\end{align*}
Since $\gamma \geq \bar{\gamma}-\delta$, we get $\Phi(u,x) \geq \bar{\Phi}(u,x)$ with
$$
\bar{\Phi}(u,x) \triangleq au^2 + ((2b-p) + px)xu + (c + p(\bar{\gamma}-\delta))x^2.
$$
Therefore, for $X$-dissipativity of $\Sigma_s$, we require  $$\bar{\Phi}(u,x) \geq 0, \quad \forall x\in [0,1], u \geq 0.$$
As $\bar{\Phi}(0,0) = 0$, $\bar{\Phi}(0,x) = (c + p(\bar{\gamma}-\delta))x^2$ and $\bar{\Phi}(u,0) = au^2$, it is clear that we require $a \geq 0$ and $(c + p(\bar{\gamma}-\delta)) \geq 0$. Now, let us define 
$$\Psi_x(u) \triangleq \bar{\Phi}(u,x) \equiv au^2 + \bar{b}u + \bar{c},$$ 
where $\bar{b} \triangleq ((2b-p) + px)x$ and $\bar{c}\triangleq (c + p(\bar{\gamma}-\delta))x^2$. Note that, in $\Psi_x(u)$, the coefficients $a \geq 0$ and $\bar{c}\geq 0$.

First, let us consider the case $a > 0$, for which $\Psi_x(u)$ is a convex quadratic function in $u$ for a given $x \in [0,1]$. Let $u^*$ be the unconstrained minimizer of $\Psi_x(u)$ for a given $x\in[0,1]$. Using the convexity and $\frac{\partial \Psi_x(u)}{\partial u} = 0$, we get 
$$u^* = -\frac{\bar{b}}{2a}\quad \mbox{ and }\quad 
\Psi_x(u^*) = -\frac{\bar{b}^2}{4a} + \bar{c}.$$
Note that, if $u^* \leq 0$, the convexity of $\Psi_x(u)$ and $\Psi_x(0) \geq 0$ automatically implies that $\Psi_x(u) \geq 0, \forall u \geq 0$. On the other hand, if $u^* \geq 0$, we then require enforcing $\Psi_x(u^*) \geq 0$ to guarantee $\Psi_x(u) \geq 0, \forall u \geq 0$. Using the $u^*$ and $\Psi_x(u^*)$ expressions, these requirements translate to: If $a>0$, then $\bar{b} \geq 0$, or $\bar{b} \leq 0$ and $4a\bar{c}-\bar{b}^2  \geq 0$.

Second, let us consider the case $a = 0$, for which we have $\Psi_x(u) = \bar{b}u + \bar{c}$ (with $\bar{c}\geq 0$). Therefore, if $a=0$, we require $\bar{b} \geq 0$ to ensure $\Psi_x(u) \geq 0, \forall u \geq 0$. Combining these two cases, for $X$-dissipativity of $\Sigma_s$, we require $\bar{c} \geq 0$ and 
% $$
% (a>0,\bar{b} \geq 0)||(a>0,\bar{b}<0,4a\bar{c}-\bar{b}^2\geq 0)||(a=0, \bar{b}\geq 0)
% $$
$$
\begin{cases}
a\geq 0, \bar{b} \geq 0, \quad \mbox{ or }\\
a>0,\bar{b}\leq 0,4a\bar{c}-\bar{b}^2\geq 0.
\end{cases}
$$

Next, recall that $\bar{b} = ((2b-p) + px)x = px^2 + (2b-p)x$, where $p>0$ and $x\in [0,1]$. Therefore, $\bar{b} \geq 0 \iff (2b-p) \geq 0$, and  $\bar{b} \leq 0 \iff b \leq 0$. Further, 
\begin{align*}
4a\bar{c} - \bar{b}^2 =&\ 4a(c + p(\bar{\gamma}-\delta))x^2 - ((2b-p) + px)^2x^2 \geq 0\\
&\iff 4a(c + p(\bar{\gamma}-\delta)) - ((2b-p) + px)^2 \geq 0\\
&\iff a(c + p(\bar{\gamma}-\delta)) - b^2 \geq 0.
\end{align*}
Furthermore, 
\begin{align*}
    \bar{c} = (c + p(\bar{\gamma}-\delta))x^2 \geq 0 \iff (c + p(\bar{\gamma}-\delta)) \geq 0.
\end{align*}

Finally, using the above equivalences, the identified requirements for $X$-dissipativity of $\Sigma_s$ can be restated as $(c + p(\bar{\gamma}-\delta)) \geq 0$ and 
\begin{equation*}
\begin{cases}
a \geq 0, (2b-p)\geq 0, \quad \mbox{or}\\
a > 0, b\leq 0, a(c + p(\bar{\gamma}-\delta)) - b^2 \geq 0,
\end{cases}
\end{equation*}
which leads to \eqref{Eq:Prop:ScalarSystem} using the Schur complement.
\end{proof}

Based on Prop. \ref{Prop:ScalarSystem} and Rm. \ref{Rm:X-DissipativityVersions}, the following corollary can be established regarding the passivity of the system $\Sigma_s$ \eqref{Eq:ScalarSystem}.

\begin{corollary}
The system $\Sigma_s$ \eqref{Eq:ScalarSystem} with $\bar{\gamma}-\delta \geq 0$ is passive, and also IF-OFP($\nu,\rho$) with $\nu \leq 0, \rho \leq p(\bar{\gamma}-\delta)$ and $p\leq 1$.
\end{corollary}
\begin{proof}
The proof follows from applying Prop. \ref{Prop:ScalarSystem} for the first two $X$-dissipativity cases given in Rm. \ref{Rm:X-DissipativityVersions}, i.e., when  
$$
X = \bm{a & b \\ b & c} = \bm{0 & \frac{1}{2} \\ \frac{1}{2} & 0} 
\mbox{ and }
X = \bm{a & b \\ b & c} = \bm{-\nu & \frac{1}{2} \\ \frac{1}{2} & -\rho}. 
$$
\end{proof}

% \begin{proof}

% Recall that:
% \begin{equation}
% \begin{aligned}
% &\mbox{Find: }\ p,\ a,\ b,\ c,\  \\ &\mbox{Sub. to: }\ p > 0,\\
% &\begin{cases}
% a \geq 0,\ (2b-p)\geq 0,\ (c + p(\bar{\gamma}-\delta)) \geq 0,\ \mbox{or}\\
% a > 0,\ b\leq 0,\ \bm{a & b\\ b & c + p(\bar{\gamma}-\delta)} \geq 0.
% \end{cases}
% \end{aligned}
% \end{equation}

% For passivity we apply $a = c = 0, b = 0.5$, leading to the conditions:  
% \begin{equation}
% \begin{aligned}
% &\mbox{Find: }\ p, \ \mbox{Sub. to: }\ p > 0, 1\geq p,\ p(\bar{\gamma}-\delta) \geq 0.
% \end{aligned}
% \end{equation}

% For IF-OFP, we apply $a = -\nu, c = -\rho, b = 0.5$, leading to the conditions:
% \begin{equation}
% \begin{aligned}
% \mbox{Find: }\ p,\ \nu,\ \rho,\  \mbox{Sub. to: }\ p > 0,\ \nu \leq 0,\ 1\geq p,\  p(\bar{\gamma}-\delta) \geq \rho
% \end{aligned}
% \end{equation}

% For L2G we apply $a = \gamma^2, c = -1, b = 0,$ which leads to the conditions

% \begin{equation}
% \begin{aligned}
% \mbox{Find: }\ p,\ \gamma^2, \mbox{Sub. to: }\ p > 0, \gamma^2 > 0, 
% \end{aligned}
% \end{equation}
% \end{proof}

\section{Problem Formulation}
\label{Sec:Formulation}

\subsection{Spreading Network Model}

We consider an epidemic spreading over a network $\Sigma$ consisting of $N \in \N$ groups denoted by $\{\Sigma_i: i\in\N_N\}$. Each group $\Sigma_i, i\in\N_N$ is represented by an open graph $\mathcal{G}_i \equiv (\mathcal{V}_i,\mathcal{E}_i)$, where vertices $\mathcal{V}_i \triangleq \{\Sigma_{i,k}:k\in\N_{N_i}\}$ represent the nodes, edges $\mathcal{E}_i \subseteq \mathcal{V}_i \times \mathcal{V}_i$ represent the epidemic spreading interactions between the nodes within the group $\Sigma_i$, and $N_i$ is the total number of nodes in $\Sigma_i$.

For a group $\Sigma_i, i\in\N_N$, an intra-group transmission matrix $M_{ii} \triangleq [M_{ii,kl}]_{k,l\in\N_{N_i}}\in\R^{N_i\times N_i}_{\geq 0}$ is used to contain all the transmission rate values corresponding to the node interactions within the group. In particular, having a positive transmission rate value $M_{ii,kl}>0$ implies the existence of a directed edge $(\Sigma_{i,l},\Sigma_{i,k}) \in \mathcal{V}_i$, indicating that node $\Sigma_{i,l}$ can infect the node $\Sigma_{i,k}$, where $k,l \in \N_{N_i}$ (allowing self-infections, i.e., $k=l$) and $i\in \N_N$. Clearly, the intra-group transmission matrix $M_{ii}$ fully defines (implies) the epidemic spreading interactions $\mathcal{V}_i$ within the group $\Sigma_i, i\in\N_N$.

Extending this notion of intra-group transmission matrices, we define inter-group transmission matrices as those that contain transmission rate values corresponding to node interactions across different groups. For example, the inter-group transmission matrix $M_{ij} \triangleq [M_{ij,kl}]_{k\in\N_{N_i},l\in\N_{N_j}}\in\R^{N_i\times N_j}_{\geq 0}$, represents the spreading interactions directed from the nodes $\{\Sigma_{j,l}, l\in\N_{N_j}\}$ in the group $\Sigma_{j}$ towards the nodes $\{\Sigma_{i,k}:k\in\N_{N_i}\}$ in the group $\Sigma_i$, where $i,j \in\N_N$ with $i\neq j$. In particular, having a positive transmission rate value $M_{ij,kl}>0$ implies the existence of a global directed edge $(\Sigma_{j,l},\Sigma_{i,k})$, indicating that node $\Sigma_{j,l}$ can infect the node $\Sigma_{i,k}$, where $k \in \N_{N_i}, l\in\N_{N_j}$ and $i,j\in \N_N$ with $i\neq j$.

We use the network Susceptible-Infected-Susceptible (SIS) model (inspired by \cite{SheHale2024}, but with added uncertainties and disturbances) to represent the dynamics of the infected proportion (also called the ``state'') $x_{i,k}(t)\in[0,1]$ in each node $\Sigma_{i,k},\, k\in\N_{N_i}, i\in\N_N$, as 
\begin{equation}
\label{Eq:NodeDynamics}
% \Sigma_{i,k}: \Big\{
\dot{x}_{i,k}(t) = -\gamma_{i,k}(t) \, x_{i,k}(t) \,+\, (1-x_{i,k}(t))u_{i,k}(t),
\end{equation}
where $\gamma_{i,k}(t)$ is the total recovery rate and $u_{i,k}(t)$ is the external infection effect (also called the ``input''). In particular, $\gamma_{i,k}(t) \triangleq \bar{\gamma}_{i,k} + \tilde{\gamma}_{i,k}(t)$, where $\bar{\gamma}_{i,k}$ is the mean recovery rate and  $\tilde{\gamma}_{i,k}(t)$ accounts for uncertainties due to modeling errors in the recovery rate such that $\vert \tilde{\gamma}_{i,k}(t) \vert \leq \delta_{i,k}$. Note that $\bar{\gamma}_{i,k}$ and $\delta_{i,k}$ are assumed as known and fixed, and satisfies $\bar{\gamma}_{i,k} > \delta_{i,k}$. On the other hand, $u_{i,k}(t)$ in \eqref{Eq:NodeDynamics} is defined as 
\begin{equation}
\label{Eq:NodeInput}
u_{i,k}(t) \triangleq \sum_{\substack{j\in\N_N}}
\sum_{l\in\N_{N_j}} M_{ij,kl}x_{j,l}(t) + w_{i,k}(t),
\end{equation}
where $w_{i,k}(t)$ is a disturbance input that accounts for both modeling errors in transmission rates and unmodeled exogenous infection effects at the node $\Sigma_{i,k}, k\in\N_{N_i}, i\in\N_N$.

\subsection{Hierarchical Networked System Representation}

We propose to represent the considered spreading network $\Sigma$ \eqref{Eq:NodeDynamics}-\eqref{Eq:NodeInput} as a hierarchy of networked systems (i.e., a hierarchical networked system), as shown in Fig. \ref{Fig:SpreadingNetwork}.

To see this, we first split the input $u_{i,k}(t),\, k\in\N_{N_i}, i\in\N_N$ in \eqref{Eq:NodeInput} into local and global components, respectively denoted by $u_{i,k}^L(t)$ and $u_{i,k}^G(t)$, as   
$$
u_{i,k}(t) = u_{i,k}^L(t) + u_{i,k}^G(t),
$$
where $u_{i,k}^L(t)$ is defined using the intra-group interactions as
\begin{equation}\label{Eq:NodeInputLocal}
u_{i,k}^L(t) \triangleq \sum_{l\in\N_{N_i}} M_{ii,kl} x_{i,l}(t), 
\end{equation}
and $u_{i,k}^G(t)$ is defined using the inter-group interactions as
\begin{equation}\label{Eq:NodeInputGlobal}
u_{i,k}^G (t) \triangleq  
\sum_{\substack{j\in\N_N \\ j \neq i }}
\sum_{l\in\N_{N_j}} M_{ij,kl}x_{j,l}(t) + w_{i,k}(t).   
\end{equation}
We next vectorize \eqref{Eq:NodeInputLocal} over $k,l\in\N_{N_i}$ to obtain
\begin{equation}\label{Eq:GroupInputLocal}
    u_i^L(t) = M_{ii} x_i(t),
\end{equation}
where $u_i^L(t) \triangleq [u_{i,k}^L(t)]_{k\in\N_{N_i}}^\T$ and $x_i(t)\triangleq [x_{i,k}]_{k\in\N_{N_i}}^\T$ are the vectorized local input and state of the group $\Sigma_i,i\in\N_N$, respectively. We then  vectorize \eqref{Eq:NodeInputGlobal} over $k,l\in\N_{N_i}$ to obtain 
\begin{equation}\label{Eq:GroupInputGlobal}
    u_i^G(t) = \sum_{\substack{j\in\N_N\\j \neq i}} M_{ij} x_j(t) + w_i(t),
\end{equation}
where $u_i^G(t) \triangleq [u_{i,k}^G(t)]_{k\in\N_{N_i}}^\T$ and $w_i(t)\triangleq [w_{i,k}(t)]_{k\in\N_{N_i}}^\T$ are the vectorized global input and disturbance of the group $\Sigma_i,i\in\N_N$, respectively. We also define 
\begin{equation}\label{Eq:GroupInput}
u_i(t) \triangleq u_i^L(t) + u_i^G(t) = M_{ii} x_i(t) +  u_i^G(t)
\end{equation} 
as the vectorized total input of the group $\Sigma_i,i\in\N_N$. 

Inspired by \eqref{Eq:GroupInput}, each group $\Sigma_i,i\in\N_N$ now can be seen as a networked system $\Sigma_i:u_i^G(t) \rightarrow x_i(t)$ (see Fig. \ref{Fig:NetworkedSystem}), comprised of the nodes $\{\Sigma_{i,k}: k\in\N_{N_i}\}$ interconnected according to the interconnection matrix 
\begin{equation}\label{Eq:IntraGroupInterconnections}
\textit{\textbf{M}}_i = \bm{M_{ii} & \I \\ \I & \0},    
\end{equation}
where the external input is $u_i^G(t)$ and the output is $x_i(t)$.

Next, we continue vectorizing \eqref{Eq:GroupInputGlobal} over $i\in\N_N$, to obtain 
\begin{equation}\label{Eq:NetworkInput}
    u^G(t) = \tilde{M} x(t) + w(t),
\end{equation}
where $u^G(t) \triangleq [(u_{i}^{G}(t))^\T]_{i\in\N_{N_i}}^\T$,  $x(t)\triangleq [x_{i}^\T(t)]_{k\in\N_{N_i}}^\T$ and $w(t)\triangleq [w_{i}^\T(t)]_{k\in\N_{N_i}}^\T$ are the vectorized global input, state and disturbance, respectively, and $\tilde{M} \triangleq [\tilde{M}_{ij}]_{i,j\in\N_N}$ with $\tilde{M}_{ij} \triangleq M_{ij}\mb{1}_{\{i \neq j\}}$ is the transmission matrix representing all inter-group interactions over the spreading network.

Finally, inspired by \eqref{Eq:NetworkInput}, we now can view the spreading network $\Sigma$ as a networked system $\Sigma: w(t) \rightarrow x(t)$ (see Fig. \ref{Fig:SpreadingNetwork}, and also Fig. \ref{Fig:NetworkedSystem}), comprised of the groups $\{\Sigma_i: i\in\N_N\}$ interconnected according to the interconnection matrix 
\begin{equation}
    \label{Eq:InterGroupInterconnections}
    \textit{\textbf{M}} = \bm{\tilde{M} & \I \\ \I & \0},
\end{equation}
where the external input affecting the spreading network is $w(t)$, which represents the disturbances due to both modeling errors and unmodeled exogenous effects, and the performance output of the spreading network is $x(t)$, which represents the infected proportions across the nodes. 

\begin{figure}[!h]
    \centering
\includegraphics[width=0.65\linewidth]{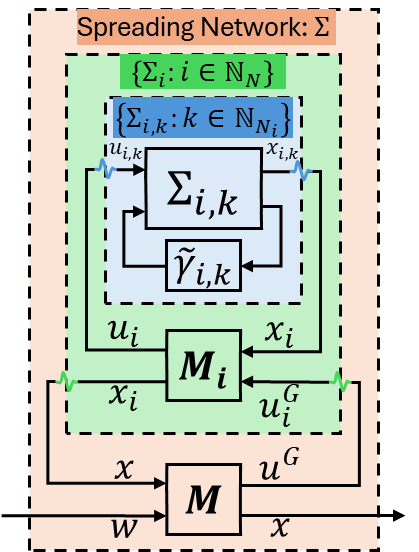}
    \caption{Spreading network representation as a hierarchical networked system.}
    \label{Fig:SpreadingNetwork}
\end{figure}

\subsection{Spreading Network Control Problem}

Given the considered spreading network model and the proposed hierarchical networked system representation, the interested spreading network control problem addressed in this paper can now be formally stated as follows. Note that here we use $\bar{M}$ to denote a known nominal value of the inter-group interconnections matrix $\tilde{M}$ in \eqref{Eq:NetworkInput}-\eqref{Eq:InterGroupInterconnections}.

\begin{problem}\label{Prob:MainProblem}
Optimally redesign the inter-group interconnections matrix $\tilde{M}$ while minimizing $J(\tilde{M}) \triangleq \Vert \tilde{M} - \bar{M} \Vert$ subject to hard constraints 
$\Vert \tilde{M} - \bar{M} \Vert_2 \leq \delta_2 \Vert \bar{M} \Vert_2$ and $- \delta_l \bar{M} \leqew \tilde{M} - \bar{M} \leqew  \delta_u \bar{M}$ (for some given $\delta_2,\delta_l,\delta_u \in [0,1]$), so that the infection-free state (i.e., $x = \0$) of the redesigned spreading network is stable, optimally $\mathfrak{X}$-dissipative (from $w$ to $x$) and scalable mesh stable. 
\end{problem}

\begin{remark}\label{Rm:DissipativityAnalysis}
As we will see in the sequel, the developed spreading network control solution to address the above problem, with some minor changes (similar in spirit to those mentioned in Rm. \ref{Rm:Analysis}), can also be used to analyze a given spreading network (i.e., to analyze stability, dissipativity and mesh stability of its infection-free state). Therefore, in this paper, we do not explicitly state the spreading network analysis problem or its solution to avoid repetition. 
\end{remark}

\begin{remark}\label{Rm:IntraGroupInterconnections}
The developed spreading network control solution to address the above problem focuses on redesigning inter-group interactions (i.e., $\tilde{M}$) while treating intra-group interactions (i.e., $M_{ii}, \forall i\in\mathcal {N}_N$) as fixed. However, as we will see in the sequel, the developed spreading network control solution, with some minor modifications, can also be used to identify optimal alterations for the intra-group interactions necessary to enforce the stability and robustness of the infection-free state for the spreading network. 
\end{remark}

\section{Dissipativity-Based Solution}
\label{Sec:Solution}

\subsection{Hierarchical Dissipativity Analysis and Design}

We begin by establishing the dissipativity properties of individual nodes in the considered spreading network. 

% For this purpose, we use Prop. \ref{Prop:ScalarSystem} and the similarity between the dynamic models used for the scalar system \eqref{Eq:ScalarSystem} and nodes \eqref{Eq:NodeDynamics}. 

\begin{lemma}\label{Lm:NodeDissipativity}
Each node $\Sigma_{i,k},\, k\in\N_{N_i},\,i\in\N_N$ \eqref{Eq:NodeDynamics} in the spreading network is $X_{i,k}$-dissipative (from $u_{i,k}(t)$ to $x_{i,k}(t)$) where $X_{i,k} \triangleq \scriptsize\bm{a & b\\ b & c}$ is given by the LMI problem:  
\begin{equation}
\label{Eq:Lm:NodeDissipativity}
\begin{aligned}
\mbox{Find: }&\ p,\ a,\ b,\ c,\  \\ 
\mbox{Sub. to: }&\ p > 0,\, a > 0,\, b\leq 0,\\
&\bm{a & b\\ b & c + p(\bar{\gamma}_{i,k}-\delta_{i,k})} \geq 0.
% \begin{cases}
% a \geq 0,\ (2b-p)\geq 0,\ (c + p(\bar{\gamma}_{i,k}-\delta_{i,k})) \geq 0,\ \mbox{or}\\
% a > 0,\ b\leq 0,\ \bm{a & b\\ b & c + p(\bar{\gamma}_{i,k}-\delta_{i,k})} \geq 0.
% \end{cases}
\end{aligned}
\end{equation}
\end{lemma}
\begin{proof}
The proof follows directly from applying Prop. \ref{Prop:ScalarSystem} for the node dynamics \eqref{Eq:NodeDynamics} as they are identical to scalar system dynamics \eqref{Eq:ScalarSystem} considered in Prop. \ref{Prop:ScalarSystem}. Note that we have omitted the $a \geq 0$ case and only considered the $a > 0$ case in \eqref{Eq:Prop:ScalarSystem} when obtaining \eqref{Eq:Lm:NodeDissipativity}. As we will see in the sequel, this (limiting to $a > 0$) helps streamline the subsequent dissipativity analysis and control stages. 
\end{proof}

Next, we use the identified node dissipativity properties in Lm. \ref{Lm:NodeDissipativity} to establish the dissipativity properties of the corresponding groups in the considered spreading network.

\begin{lemma}\label{Lm:GroupDissipativity}
Each group $\Sigma_i,\, i\in\N_N$ (see Fig. \ref{Fig:SpreadingNetwork}) in the spreading network is $\mathcal{X}_i$-dissipative (from $u_i^G(t)$ to $x_i(t)$) where $\mathcal{X}_i \triangleq [\mathcal{X}_i^{lm}]_{l,m\in\N_2}$ is given by the LMI problem:
\begin{equation}\label{Eq:Lm:GroupDissipativity}
    \begin{aligned}
    \mbox{Find: }& \mathcal{X}_i, \{p_{i,k}: k\in\mathbb{N}_{N_i}\}\\
    \mbox{Sub. to: }& p_{i,k} > 0,\ \forall k\in\mathbb{N}_{N_i},\ \Phi_i > 0,
    \end{aligned}
\end{equation}
where $\Phi_i$ takes the form
\begin{equation}
\label{Eq:Lm:GroupDissipativity2}
\Phi_i \triangleq 
\begin{bmatrix} 
\textbf{X}_{p_i}^{11} & \textbf{0} & \textbf{X}_{p_i}^{11}M_{ii} & \textbf{X}_{p_i}^{11} \\
\star & -\mathcal{X}^{22}_i & -\mathcal{X}^{22}_i & \0\\ 
\star & \star & - \mathcal{H}(\textbf{X}_{p_i}^{21}M_{ii})-\textbf{X}_{p_i}^{22} & -\textbf{X}_{p_i}^{21}+\mathcal{X}^{21}_i \\
\star & \star & \star &  \mathcal{X}^{11}_i
\end{bmatrix}
\end{equation}
and each $\textbf{X}_{p_i}^{lm}$ term for any $l,m\in\N_2$ is defined using the identified node dissipativity properties $\{X_{i,k}^{lm}:k\in\N_{N_i}\}$ in Lm. \ref{Lm:NodeDissipativity} as $\textbf{X}_{p_i}^{lm} \triangleq diag([p_{i,k}X_{i,k}^{lm}]_{k\in\N_{N_i}})$.
% $\textbf{X}_p^{11}M_{uy}\triangleq L_{uy}$.
\end{lemma}
\begin{proof}
Inspired by the networked system representation (see Fig. \ref{Fig:SpreadingNetwork}) of each group $\Sigma_i,\, i\in\N_N$, for its $\mathcal{X}_i$-dissipativity analysis, we apply Co. \ref{Co:NetworkedSystem}.
To this end, we first have to ensure that As. \ref{As:SubsystemDissipativity} required in Co. \ref{Co:NetworkedSystem} holds for each node $\Sigma_{i,k},\, k\in\N_{N_i}$ in group $\Sigma_i, i\in\N_N$. 
This is automatically satisfied by the dissipativity properties $X_{i,k}$ of each node $\Sigma_{i,k},\,k\in \N_{N_i}$ identified in Lm. \ref{Lm:NodeDissipativity} as $X_{i,k}^{11} = a > 0$ is enforced in \eqref{Eq:Lm:NodeDissipativity}.  
Note also that the assumed constraints \eqref{Eq:InterconnectionMatrix2} in Co. \ref{Co:NetworkedSystem} are consistent with the intra-group interactions structure \eqref{Eq:IntraGroupInterconnections}. 

Next, as Co. \ref{Co:NetworkedSystem} is intended for control design problems (not for analysis problems), following Rm. \ref{Rm:Analysis}, we introduce an extra constraint for the design variable $L_{uy}$ in Co. \ref{Co:NetworkedSystem} as $L_{uy} = \textbf{X}_{p}^{11}M_{uy}$.
Finally, we simplify the LMI problem \eqref{Eq:Co:NetworkedSystem} in Co. \ref{Co:NetworkedSystem}, with the appropriate substitutions: $M_{uy} = M_{ii}$, $\textbf{X}_{p}^{lm} = \textbf{X}_{p_i}^{lm}$, $\X^{lm} = \X_i^{lm}$, and $X_i^{lm} = X_{i,k}^{lm},\, \forall l,m\in\N_2$, which leads to the LMI problem given in \eqref{Eq:Lm:GroupDissipativity}. 
\end{proof}

Finally, we use the identified group dissipativity properties in Lm. \ref{Lm:GroupDissipativity} to optimally redesign the inter-group interconnections $\tilde{M}$ in \eqref{Eq:NetworkInput}-\eqref{Eq:InterGroupInterconnections} to enforce (and optimize) stability, dissipativity, and scalable mesh stability properties of the redesigned spreading network. In this pursuit, as stated in the Prob. \ref{Prob:MainProblem}, we also optimize the deviation of $\tilde{M}$ from its nominal value $\bar{M}$ subject to some given hard constraints.

\begin{theorem}\label{Th:NetworkDissipativity}
Addressing Prob. \ref{Prob:MainProblem}, the spreading network $\Sigma$ shown in Fig. \ref{Fig:SpreadingNetwork} can be made stable, optimally $\mathfrak{X}$-dissipative (where $\mathfrak{X} \triangleq [\mathfrak{X}^{lm}]_{l,m\in\N_2}$) and scalable mesh stable, by redesigning the inter-group interaction matrix $\tilde{M}$ using the LMI problem
\begin{equation}\label{Eq:Th:NetworkDissipativity}
\begin{aligned}
\min_{\substack{L, \mathfrak{X} \\ \{p_i: i\in\mathbb{N}_N\}}}\ 
&\tilde{J}(L,\mathfrak{X}) = \Vert L - \mathbfcal{X}_p^{11}\bar{M} \Vert + \bar{J}(\mathfrak{X}) \\ 
\mbox{Sub. to: }& p_i > 0,\ \forall i\in\mathbb{N}_N,\ \Phi > 0,\\
& \tilde{\lambda}_{i1} \sum_{j\in\N_N} \vert (\X_i^{11})^{-1} L^{ij} \vert < p_i,\ \forall i \in \N_N,\\  
& 0 \leq \bm{
    \I & Y^\T & \0\\
    Y & \mathbfcal{X}_p^{11}Y^\T + Y\mathbfcal{X}_p^{11}  & L-\mathbfcal{X}_p^{11}\bar{M} \\ 
    \0 & L^\T - \bar{M}^\T \mathbfcal{X}_p^{11} & \delta^2_2 \bar{M}^\T \bar{M}},\\
& - \delta_l \mathbfcal{X}_p^{11}\bar{M} \ll L - \mathbfcal{X}_p^{11}\bar{M} \ll  \delta_u \mathbfcal{X}_p^{11}\bar{M},
\end{aligned}
\end{equation}
where $\Phi$ takes the form 
\begin{equation}
\label{Eq:Th:NetworkDissipativity2}
\Phi \triangleq 
\begin{bmatrix} 
\mathbfcal{X}_p^{11} & \textbf{0} & L & \mathbfcal{X}_p^{11} \\
\star & - \mathfrak{X}^{22} & -\mathfrak{X}^{22} & \0\\ 
\star & \star & -\mathcal{H}(\mathbfcal{X}^{21}L)-\mathbfcal{X}_p^{22} & -\mathbfcal{X}_p^{21}+\mathfrak{X}^{21} \\
\star & \star & \star & \mathfrak{X}^{11}
\end{bmatrix},
\end{equation}
and each $\mathbfcal{X}_{p}^{lm}$ term for any $l,m\in\N_2$ is defined using the identified group dissipativity properties $\{\mathcal{X}_i^{lm}:i\in\N_N\}$ in Lm. \ref{Lm:GroupDissipativity} as 
$\mathbfcal{X}_{p}^{lm} \triangleq diag([p_i \mathcal{X}_{i}^{lm}]_{i\in\N_{N_i}})$, 
$\mathbfcal{X}^{12}\triangleq diag([(\mathcal{X}_i^{11})^{-1}\mathcal{X}_i^{12}]_{i\in\mathbb{N}_N})$, $\mathbfcal{X}^{21}\triangleq (\mathbfcal{X}^{12})^\top$, 
$\tilde{\lambda}_{i1} \triangleq \sqrt{\frac{\overline{\lambda}(R_i)\overline{\lambda}(P_i)}{\underline{\lambda}(P_i)\underline{\lambda}(Q_i)}}$ (where $P_i\triangleq \diag([\tilde{p}_{i,k}\bar{p}_{i,k}]_{k\in\N_{N_i}})$ with $\tilde{p}_{i,k} \triangleq p_{i,k}$ form Lm. \ref{Lm:GroupDissipativity} and $\bar{p}_{i,k}\triangleq p$ values from Lm. \ref{Lm:NodeDissipativity}, and $Q_i$ and $R_i$ satisfy $(\X_i^{22}+\X_i^{21}) \leq -Q_i < 0$ and $(\X_i^{11} + \X_i^{12}) \leq R_i$), $Y \triangleq \I$ (see also Rm. \ref{Rm:BMIforY}), $\bar{J}(\mathfrak{X})$ is a convex function of $\mathfrak{X}$, $L$ shares the sparsity structure of $\tilde{M}$, and 
$\tilde{M} \triangleq (\mathbfcal{X}_p^{11})^{-1} L.$
\end{theorem}
\begin{proof}
Based on the networked system representation  \eqref{Eq:InterGroupInterconnections} of the spreading network $\Sigma$, to synthesize its interconnection matrix $\tilde{M}$ to enforce (and optimize) stability, $\mathfrak{X}$-dissipativity and scalable mesh stability, we apply Pr. \ref{Pr:NetworkedSystemSMS}. To this end, we first have to ensure that As. \ref{As:SubsystemDissipativity2} required in Pr. \ref{Pr:NetworkedSystemSMS} holds for each group $\Sigma_i,\, i\in\N_{N}$ in the spreading network. 
The first condition in As. \ref{As:SubsystemDissipativity2} is automatically satisfied by the dissipativity properties $\X_i$ of each group $\Sigma_i,\,i\in \N_N$ identified in Lm. \ref{Lm:GroupDissipativity} as $\X_i^{11} > 0$ is enforced in \eqref{Eq:Lm:GroupDissipativity}. 
For the second condition in As. \ref{As:SubsystemDissipativity2}, as the corresponding quadratic storage function $V_i(x_i) \triangleq x_i^\T P_i x_i$ for each group $\Sigma_i,\,i\in \N_N$, it is clear that we can use the given $P_i$ definition based on $p_{i,k}$ and $p$ values respectively observed in Lemmas \ref{Lm:NodeDissipativity} and \ref{Lm:GroupDissipativity}. 
Note also that the assumed constraints \eqref{Eq:InterconnectionMatrix2} in Prop. \ref{Pr:NetworkedSystemSMS} are consistent with the inter-group interactions structure \eqref{Eq:InterGroupInterconnections}. 
Consequently, starting from \eqref{Eq:Pr:NetworkedSystemSMS}, using the appropriate substitutions: $L_{uy} = L$, $M_{uy} = \tilde{M}$, $\textbf{X}_{p}^{lm} = \mathbfcal{X}_p^{lm}$, $\X^{lm} = \mathfrak{X}^{lm}$, and $X_i^{lm} = \X_i^{lm},\, \forall l,m\in\N_2$, we can obtain the first three LMI constraints in \eqref{Eq:Th:NetworkDissipativity}. 

% Problem statement addressing
Next, to enforce the soft and hard constraints on deviations of $\tilde{M}$ from its nominal value $\bar{M}$ (as required in Prob. \ref{Prob:MainProblem}), we apply Propositions \ref{Pr:AlternativeObjective} and \ref{Pr:AlternativeConstraints}. This provides the first term in the joint objective function $\tilde{J}(L,\mathfrak{X})$, and the fourth and fifth LMI constraints in \eqref{Eq:Th:NetworkDissipativity}. Note that the $\delta_2$ parameter appearing in the fourth LMI constraint in \eqref{Eq:Th:NetworkDissipativity} can also be treated as a design variable to be minimized, and the fifth LMI constraint in \eqref{Eq:Th:NetworkDissipativity} requires each $\X_i^{11}$ obtained in Lm. \ref{Lm:GroupDissipativity} to be diagonal (see Prop. \ref{Pr:AlternativeConstraints}). 
Finally, we point out that we have included the convex term $\bar{J}(\mathfrak{X})$ as the second term in $\tilde{J}(L,\mathfrak{X})$ to enable joint optimization of inter-group interconnections $\tilde{M}$ (via $L$) and the spreading network dissipativity $\mathfrak{X}$. 
\end{proof}

\subsection{Supporting Feasibility with Necessary Conditions}

Inspired by Prop. \ref{Prop:NecessaryConditions}, the following lemma presents a necessary condition to be included in each node dissipativity analysis (in Lm. \ref{Lm:NodeDissipativity}) to support the feasibility of the corresponding subsequent group dissipativity analysis (in Lm. \ref{Lm:GroupDissipativity}).

\begin{lemma}\label{Lm:NecessaryConForGroupDis}
At each group $\Sigma_i,\, i\in\N_N$, for the feasibility of its $\X_i$-dissipativity analysis problem (in Lm. \ref{Lm:GroupDissipativity}), a necessary condition that can be enforced at each node $\Sigma_{i,k}, k\in\N_{N_i}$ in its $X_{i,k}$-dissipativity analysis problem (in Lm. \ref{Lm:NodeDissipativity}, $X_{i,k} \triangleq \scriptsize\bm{a & b \\ b & c} \normalsize$) is the feasibility of the LMI problem:
\begin{equation}\label{Eq:Lm:NecessaryConForGroupDis}
\begin{aligned}
\mbox{Find: }&\ a,\, b,\, c,\, \bar{a},\, \bar{b},\, \bar{c},\\
\mbox{Sub. to: }
&\tilde{\Phi}_{i,k} \triangleq 
\begin{bmatrix} 
a & 0 & a\,m & a \\
\star & -\bar{c} & -\bar{c} & 0\\ 
\star & \star & -2b\,m - c & - b + \bar{b} \\
\star & \star & \star &  \bar{a}
\end{bmatrix} > 0,
\end{aligned}
\end{equation}
where $\scriptsize \bm{\bar{a} & \bar{b}\\ \bar{b} & \bar{c}}$ shares the sparsity structure of $[\mathcal{X}_{i,k}^{lm}]_{l,m\in\N_2}$ ($\mathcal{X}_{i,k}^{lm}$ is the $k$\tsup{th} diagonal element of $\mathcal{X}^{lm}_i$ in \eqref{Eq:Lm:GroupDissipativity}), and $m \triangleq M_{ii,kk}$ is the $k$-th diagonal element of $M_{ii}$. 
\end{lemma}
\begin{proof}
The proof follows directly from applying Prop. \ref{Prop:NecessaryConditions} for the networked system representation \eqref{Eq:GroupInput}-\eqref{Eq:IntraGroupInterconnections} of each group $\Sigma_i,\, i\in\N_N$, particularly for for each group's $\mathcal{X}_i$-dissipativity analysis (in Lm. \ref{Lm:GroupDissipativity}). 
\end{proof}

Similar to Lm. \ref{Lm:NecessaryConForGroupDis}, inspired by Prop. \ref{Prop:NecessaryConditions}, the following lemma presents a necessary condition to be included in each group dissipativity analysis (in Lm. \ref{Lm:GroupDissipativity}) to support the feasibility of the subsequent spreading network design problem (in Th. \ref{Th:NetworkDissipativity}). 

\begin{lemma}\label{Lm:NecessaryConForGroupDis}
For the feasibility of $\mathfrak{X}$-dissipative spreading network design problem (in Lm. \ref{Th:NetworkDissipativity}), a necessary condition that can be enforced at each group $\Sigma_i,\, i\in\N_{N}$ in its $\mathcal{X}_i$-dissipativity analysis problem (in Lm. \ref{Lm:GroupDissipativity}) is the feasibility of the LMI problem:
\begin{equation}\label{Eq:Lm:NecessaryConForNetworkDis}
\begin{aligned}
\mbox{Find: }&\ \mathcal{X}_i, \bar{\mathfrak{X}}_i\\
\mbox{Sub. to: }&\ 
\tilde{\Phi}_i \triangleq 
\begin{bmatrix} 
\mathcal{X}_i^{11} & \0 & \0 & \mathcal{X}_i^{11} \\
\star & -\bar{\mathfrak{X}}^{22}_i & -\bar{\mathfrak{X}}^{22}_i & \0\\ 
\star & \star &  - \mathcal{X}_i^{22} & -\mathcal{X}_i^{21}+\bar{\mathfrak{X}}^{21}_i \\
\star & \star & \star & \bar{\mathfrak{X}}^{11}_i
\end{bmatrix} > 0,
\end{aligned}
\end{equation}
where $\bar{\mathfrak{X}}_i \triangleq [\bar{\mathfrak{X}}_{i}^{lm}]_{l,m\in\N_2}$ shares the sparsity structure of $[\mathfrak{X}_{i}^{lm}]_{l,m\in\N_2}$ ($\mathfrak{X}_{i}^{lm}$ is the $i$\tsup{th} diagonal block of $\mathfrak{X}^{lm}$ in \eqref{Eq:Th:NetworkDissipativity}).
% Moreover, a necessary condition that can be enforced at each node $\Sigma_{i,k}, k\in\N_{N_i}, i\in\N_N$ in its $X_{i,k}$-dissipativity analysis problem (in Lm. \ref{Lm:NodeDissipativity}, $X_{i,k} \triangleq \scriptsize \bm{a & b \\ b & c}$) is the feasibility of the LMI problem: 
% \begin{equation}\label{Eq:Lm:NecessaryConForNetworkDis3}
% \begin{aligned}
% \mbox{Find: }&\ a,\, b,\, c,\, \bar{a},\, \bar{b},\, \bar{c},\\
% \mbox{Sub. to: }&\ \tilde{\tilde{\Phi}}_{i,k} > 0,
% \end{aligned}
% \end{equation}
% where $\scriptsize \bm{\bar{\bar{a}} & \bar{\bar{b}} \\ \bar{\bar{b}} & \bar{\bar{c}}}$ shadows $[\mathfrak{X}_{i,k}^{lm}]_{l,m\in\N_2}$ (recall that each $\mathfrak{X}_{i,k}^{lm}$ is the $k$\tsup{th} diagonal element of $\mathfrak{X}^{lm}_i$ in \eqref{Eq:Lm:GroupDissipativity}), 
% \begin{equation}
% \label{Eq:Lm:NecessaryConForNetworkDis4}
% \tilde{\tilde{\Phi}}_{i,k} \triangleq 
% \begin{bmatrix} 
% a & 0 & 0 & a \\
% 0 & -\bar{\bar{c}} & -\bar{\bar{c}} & 0\\ 
% 0 & -\bar{\bar{c}} & - c & - b + \bar{\bar{b}} \\
% a & 0 & -b + \bar{\bar{b}} &  \bar{\bar{a}}
% \end{bmatrix}.
% \end{equation}
\end{lemma}
\begin{proof}
The proof follows directly from applying Prop. \ref{Prop:NecessaryConditions} for the networked system representation \eqref{Eq:NetworkInput}-\eqref{Eq:InterGroupInterconnections} of the spreading network, particularly for its $\mathfrak{X}$-dissipative spreading network design (in Th. \ref{Th:NetworkDissipativity}). Note that, here we have used the fact that $\tilde{M}_{ii} = \0, i\in\N_N$ to further simplify the derived necessary condition and obtain \eqref{Eq:Lm:NecessaryConForNetworkDis}.
\end{proof}

%%%% Revised up to this

\begin{figure*}[!b]
\hrulefill
\centering
\begin{equation}
\label{Eq:Th:NetworkDissipativity2}
\Phi \triangleq 
\begin{bmatrix} 
\mathbfcal{X}_p^{11} & \textbf{0} & L & \mathbfcal{X}_p^{11} \\
\star & -\mathfrak{X}^{22} & -\mathfrak{X}^{22} & \0\\ 
\star & \star & -\mathcal{H}(\mathbfcal{X}^{21}L)-\mathbfcal{X}_p^{22} & -\mathbfcal{X}_p^{21}+\mathfrak{X}^{21} \\
\star & \star & \star & \mathfrak{X}^{11}
\end{bmatrix}
\end{equation}
\end{figure*}

\subsection{Overall Spreading Network Control Solution}

Combining the above theoretical results, we provide the following algorithm that summarizes the proposed hierarchical spreading network control solution, i.e., the complete spreading network design process that addresses the considered Problem \ref{Prob:MainProblem} in this paper. The accompanying theorem summarizes the overall outcome. 

\begin{algorithm}[!h]
\caption{Proposed Spreading Network Design Process}\label{Alg:NetworkDesign}
\begin{algorithmic}[1]
\State \textbf{Input: } 
$\{(\bar{\gamma}_{i,k},\delta_{i,k}): k\in\N_{N_i}, i\in\N_N\}$, $\bar{M}$ and $\{M_{ii}:i\in\N_N\}$.  
\For{$i=1,2,\ldots,N$} \Comment{Analyzing Groups}
\For{$k=1,2,\ldots,N_i$} \Comment{Analyzing Nodes}
\State Find $X_{i,k}$ by solving \eqref{Eq:Lm:NodeDissipativity} with $\tilde{\Phi}_{i,k}>0$ in \eqref{Eq:Lm:NecessaryConForGroupDis}.
\EndFor
\State Find $\mathcal{X}_i$ by solving \eqref{Eq:Lm:GroupDissipativity} with 
$\tilde{\Phi}_{i}>0$ in \eqref{Eq:Lm:NecessaryConForNetworkDis}.
\EndFor
\State Find $\tilde{M}$, $\mathfrak{X}$ by solving \eqref{Eq:Th:NetworkDissipativity}. \Comment{Designing the Network}
\end{algorithmic}
\end{algorithm}

\begin{theorem}
Given the node recovery rate information, intra-group interactions, and nominal inter-group interactions, addressing Prob. \ref{Prob:MainProblem}, we can enforce the spreading network $\Sigma$ shown in Fig. \ref{Fig:SpreadingNetwork} to be stable, optimally $\mathfrak{X}$-dissipative (where $\mathfrak{X} \triangleq [\mathfrak{X}^{lm}]_{l,m\in\N_2}$), and scalable mesh stable, by designing the interconnection matrix $\tilde{M}$ following the design process outlined in Alg. \ref{Alg:NetworkDesign}. 
\end{theorem}

\begin{proof}
The proof is complete by noticing that the LMI problem pairs  \{\eqref{Eq:Lm:NodeDissipativity},\eqref{Eq:Lm:NecessaryConForGroupDis}\} (at each node) and \{\eqref{Eq:Lm:GroupDissipativity}, \eqref{Eq:Lm:NecessaryConForNetworkDis}\} (at each group) are compatible with each other so that they can be respectively combined as stated in Step 4 and Step 6 of the Alg. \ref{Alg:NetworkDesign}. Finally, upon solving these combined LMI problems at all nodes and groups, we obtain the necessary (and supportive) group dissipativity information to solve the main spreading network design problem \eqref{Eq:Th:NetworkDissipativity} (established in Th. \ref{Th:NetworkDissipativity}) as stated in Step 8 of the Alg. \ref{Alg:NetworkDesign}. 
\end{proof}

\begin{remark}\label{Rm:DissipativityAnalysis2}
As stated in Rm. \ref{Rm:DissipativityAnalysis}, for the purpose of spreading network analysis (e.g., of its stability, dissipativity, and/or mesh-stability) under a given inter-group interconnections $\tilde{M}=\bar{M}$, we can use the proposed spreading network design process (i.e., Alg. \ref{Alg:NetworkDesign}), after some minor changes. One approach is to set $\delta_l = \delta_u = \delta_2 = 0$ in \eqref{Eq:Th:NetworkDissipativity}. Alternatively, the design variable $L$ can be removed entirely from \eqref{Eq:Th:NetworkDissipativity} by substituting $L = \X_p^{11}\bar{M}$. Subsequently, the feasibility (or infeasibility) of this revised Alg. \ref{Alg:NetworkDesign} will imply the success (or inconclusiveness) of the analysis. It is worth noting that the identified supporting conditions \eqref{Eq:Lm:NecessaryConForGroupDis} and \ref{Eq:Lm:NecessaryConForNetworkDis} included in Alg. \ref{Alg:NetworkDesign} significantly reduce the possibility for infeasibility in both analysis and design of spreading networks.   
\end{remark}
% (albeit they can be inconclusive)

\begin{remark}\label{Rm:ExtendingBeyond}
The proposed framework's flexibility to conduct both analysis and control (design) at different levels of the hierarchical spreading network is highly beneficial. For example, it allows us to redesign intra-group interconnections if a specific group is not sufficiently dissipative (thus compromising the overall performance of the spreading network). Similarly, it allows us to provide recommendations for improving node characteristics (e.g., using localized pharmaceutical intervention method) to achieve better node-level dissipativity properties that support improving group- and network-level dissipativity properties. Moreover, the same flexibility will play a major role when there are higher layers beyond the network level (e.g., a global network layer) in the considered hierarchical spreading network model.
\end{remark}

%%%% Revised up to this
\section{Simulation Results}
\label{Sec:Simulation}

To provide operational details and showcase the benefits of the proposed dissipativity-based spreading network design technique, we implement it along with several of its variants and two standard techniques for a randomly generated spreading network. In particular, we report the variation of the average infection level observed under different design methods over an extended period, where different disturbances were injected into the spreading network. We also report the percentage reductions in the inter-group interconnections exercised by different design methods to measure the exerted spreading network design effort.

% Considered network 
\subsection{Considered Spreading Network}
As shown in Fig. \ref{Fig:1a}-\ref{Fig:1b}, the considered spreading network consists of four groups $G_1,G_2,G_3,G_4$ with $5,6,7,4$ nodes, respectively. The mean recovery rates (i.e., $\bar{\gamma}_{i,k}$) for the nodes were selected from the uniform random distribution in the interval $[0.4, 0.9]$. The recovery rate deviation bound (i.e., $\delta_{i,k}$) was set as $5\%$ of the selected recovery rate at each node. The initial infection level for each node was selected from the uniform random distribution in the interval $[0, 1]$.

% Interconnections
The intra-group interconnections were created randomly with a probability of $0.3$ for any node pair in a group, and the interconnection levels were selected from the uniform random distribution in the interval $[0.1, 0.4]$. While we made no assumption regarding the group connectivity (like strong connectivity), we scaled down the intra-group interconnections to ensure group stability of the infection-free state when isolated from other groups. This scaling was done only to create an ideal reference behavior to compare with other spreading network design methods. On the other hand, the inter-group interconnections were created randomly with a probability of $0.2$ for any node pair in the spreading network, and the interconnection levels were selected from the uniform random distribution in the internal $[0.1,0.4]$ without any scaling.

% Disturbances
The disturbance input (i.e., $w_{i,k}(t)$) affecting each node was chosen to include four distinct components, each active during specific time windows and triggered asynchronously at different nodes: a brief, low-amplitude sinusoidal fluctuation during $t\in[40, 45]$, a step offset during $t\in[80,85]$, a toggling component during $t\in[120,130]$, and mild random noise added during $t\in[0,160]$. These short-lived perturbations combine to create a rich disturbance profile to evaluate the different design approaches. 

% Uncontrolled/full controlled behavior
We simulated this spreading network under different inter-group interconnection configurations given by different spreading network design methods, including the two cases where we used none of and all of the original inter-group interconnections. The observed variations of the average infection level (fraction) of the spreading network over a period $t\in[0,200]$ are shown in Fig. \ref{Fig:2}. Notice that when no inter-group interconnections are used, by design, the spreading network stabilizes to the infection-free state very rapidly overcoming all disturbance conditions. However, when there are inter-group interconnections, the spreading network does not stabilize to the infection-free state. Hence these two average infections profiles provide the best and worst case scenarios for the spreading process to compare other spreading network design methods.

\subsection{Dissipativity-Based Design}
% Dissipativity based design
We implemented the proposed dissipativity-based spreading network design method as follows. First, each node $\Sigma_{i,k}, k\in\N_{N_i}, i\in\N_N$ was considered to be IF-OFP($\nu,\rho$) (see Rm. \ref{Rm:X-DissipativityVersions}), and hence, the problem $\mathbb{P}_{i,k}$ was solved by setting $b=0.5$ and optimizing for decision variables $a$ and $c$, representative of the passivity indices $-\nu$ and $-\rho$, respectively. As we plan to establish similar IF-OFP properties at the group level, in problem $\mathbb{P}_{i,k}$, we also set $\bar{b}=0.5$ and included $\bar{a}$ and $\bar{c}$ as decision variables. Overall, when solving the problem $\mathbb{P}_{i,k}$, we optimize the objective function $J_{i,k} \triangleq a+c+\bar{a}+\bar{c}$ to identify the maximum feasible node passivity indices while helping to establish similarly high passivity indices at the group level. 

Next, each group $\Sigma_i, i\in\N_N$ was considered to have a vectored IF-OFP property, and hence the problem $\mathbb{P}_{i}$ was solved by setting $\mathcal{X}_i^{12} = 0.5\I$ and optimizing for diagonal matrices (decision variables) $\mathcal{X}_i^{11}$ and $\mathcal{X}_i^{22}$. As we plan to establish the network L2G($\gamma$) (see Rm. \ref{Rm:X-DissipativityVersions}), in problem $\mathbb{P}_{i}$, we also set $\bar{\mathfrak{X}}_i^{12}=\0$, $\bar{\mathfrak{X}}_i^{22}=-\I$ and $\bar{\mathfrak{X}}_i^{11}=\bar{\gamma}_i\I$ and include $\bar{\gamma}_i$ as a decision variable. Overall, when solving the problem $\mathbb{P}_i$, we optimize the objective function $J_{i} \triangleq \mbox{trace}(\mathcal{X}_i^{11})+\mbox{trace}(\mathcal{X}_i^{22})+\bar{\gamma}_i$ to identify maximum feasible group passivity indices while helping to achieve lowest possible $L_2$ gain values for the spreading network. 

Finally, we set out to design the inter-group interconnections $\tilde{M}$ in the spreading network $\Sigma$ to make it L2G($\gamma$), i.e., finite-gain $L_2$ stable with a gain $\gamma$. 
For this, the problem $\mathbb{P}$ was solved by setting $\mathfrak{X}^{12}=\0$, $\mathfrak{X}^{22}=-\I$ and $\mathfrak{X}^{11} = \bar{\gamma}\I$, and optimizing for $\tilde{M}$ and $\bar{\gamma}$ as decision variables. In particular, when solving problem $\mathbb{P}$, we optimized the objective function:
$$
J = c_M \Vert L - \mathbfcal{X}_p^{11}\bar{M} \Vert_1 + \bar{\gamma}
$$
subject to a limit $\delta_M$ on the fractional deviation allowed for the magnitude of any inter-group interconnection level, i.e.,  
$$
(1-\delta_M)\bar{M} << \tilde{M}  \iff (1-\delta_M) \mathbfcal{X}_p^{11} \bar{M} << L.   
$$
In the baseline dissipativity-based design approach, we used the problem parameters $c_M = 1$ and $\delta_M = 1$. For comparison purposes, we also evaluate the dissipativity-based design for $c_M \in \{10^{-9}, 10^9\}$ (effectively, $c_M \in \{0, \infty\}$) and $\delta_M \in \{0.9, 0.95\}$. Henceforth, these dissipativity-based methods are denoted as DissBC($c_M,\delta_M$), e.g., the baseline dissipativity-based method is denoted as DisBC($1,1$).

\subsection{Threshold and Degree Based Methods}
To compare the performance of the dissipativity-based spreading network designs, we also examine two basic, practical spreading network design methods (see also \cite{Yi2022}). The first approach prunes inter-group interconnections with weights beyond a given threshold $t_M$. By removing the strongest inter-group links, this method reduces the most impactful pathways for disease propagation.  Henceforth, this threshold-based approach is denoted as TBC($t_M$). The second approach identifies the top $d_M$ fraction of the nodes in terms of their inter-group out-degree and removes all their outgoing connections to nodes in other groups. By isolating these highly connected nodes, this method limits the spread of epidemics through critical inter-group links. Henceforth, this degree-based approach is denoted as DegBC($t_M$).

% How to mke it fair
To make the comparison of different spreading network design methods fair, we first identified a ``design effort'' metric:
$$
J_M \triangleq \sum_{\substack{\forall i,j \in\N_N, i\neq j,\\ \forall k\in\N_{N_i},l\in\N_{N_j}}} \frac{\bar{M}^{ij,kl} - \tilde{M}^{ij,kl}}{\bar{M}^{ij,kl}}
$$
representing the average fractional cutdown of the interconnections, observed under each method. The baseline DissBC($1,1$) method observed an average fractional cutdown of $J_c = 87.1\%$, and hence the parameters $t_M$ and $d_M$ respectively of TBC($t_M$) and DegBC($d_M$) methods were tuned to achieve the similar $J_M$ level. The resulting parameters were $t_M = 0.18$ and $d_M = 0.72$.

\subsection{Results}

Based on the reported average infection levels in Fig. \ref{Fig:2}, the baseline DissBC($1,1$) method provided the closest behavior to the best case (i.e., the case without interconnections). The TBC($0.18$) and DegBC($0.72$) methods while achieved stability, were considerably slow to respond to various disturbances. This effectiveness of the proposed DissBC($1,1$) method can be attributed to the systematic dynamics-aware LMI-based design process, which optimally tailors the inter-group interconnections while optimizing the a robust stability measure.

% Performance measures
As stated in Rm. \ref{Rm:DissipativityAnalysis}, the proposed dissipativity-based spreading networks design approach can also be used to analyze the dissipativity of a given spreading network. Consequently, it can be used to analyze the $L_2$ gain level $\gamma$ of any given spreading network. We used this $\gamma$ level as a performance metric besides the time-averaged infection level (fraction) 
$$J_x \triangleq  \frac{1}{T}\int_0^T x(t)dt,$$ 
and the average interconnection reductions fraction $J_M$ as measures of performance when comparing different spreading network design methods. The observed performance metrics are summarized in Tab. \ref{Tab:Results}. 

%%%% Revised up to this
Based on Tab. \ref{Tab:Results}, it is clear that the dissipativity-based spreading network design have superior (i.e., lower) $L_2$ gain $\gamma$ levels compared to other methods. In particular, DissBC($10^-9,1$) (i.e., when $c_M = 10^{-9}$ and thus, when the design objective is effectively $\gamma$) has obtained the lowest $\gamma$ value, at $\gamma = 24.10$. Interestingly, as can be seen in Figs. \ref{Fig:4a} and \ref{Fig:4c}, the average reduction of interconnections in this DissBC($10^{-9},1$) case is also the lowest, at $J_M = 0.8168$. This implies that it is not necessary to reduce interconnections significantly to obtain guaranteed theoretical robust stability measures. 

Based on both Fig. \ref{Fig:2} and Tab. \ref{Tab:Results}, dissipativity based spreading network design methods also have superior experimental average infections level performance (i.e., lower $J_x$ metrics). In particular, DissBC($10^9,1$) (i.e., when $c_M = 10^{9}$ and thus when the design objective is effectively the interconnections) has obtained the lowered $J_x$ value, at $J_x = 0.0638$. This is because, due to its focus on optimizing only the interconnections, the solution has drastically cutdown interconnections. This behavior is evident from Figs. \ref{Fig:4b} and \ref{Fig:4d} and also from the fact that this method achieves the highest $J_M$ values, at $J_M = 0.9575$. 

As shown in Fig. \ref{Fig:3}, in dissipativity based designs, by reducing the $\delta_M$ parameter from $1$ to $0.95$ and $0.9$ (enforcing stricter constraints on interconnection deviations), DissBC($1,0.95$) and Diss($1,0.9$) methods have sacrificed both experimental performance $J_x$ and theoretical performance $\gamma$. 

Considering these behaviors, it is clear that the selected baseline dissipativity-based design approach DissBC($1,1$) considers a moderate case where both the robust stability measure $\gamma$ and interconnection cost are treated equally, without forced constraints of interconnection deviations. As can be seen in Fig. \ref{Fig:1c} and \ref{Fig:1d}, this leads to a spreading network with a close to the best: (i) robust stability guarantee with $\gamma = 24.15$, (ii) experimental performance $J_x = 0.0380$, and (iii) interconnection reductions $J_M = 0.8710$ highlighting the benefit of co-designing networks while using a well-balanced multiple objective function.

% Conclusion
Overall, the results indicate that the proposed dissipativity-based approach provides a structured and efficient means of spreading networks. Its robustness against asynchronous disturbances while using specifically designed minimal alterations to the nominal network further highlights its suitability for practical applications where system dynamics are complex and uncertain.

\begin{figure*}[!th]
\centering
\subfloat[Intra-group interconnections]{
\includegraphics[width=0.46\columnwidth]{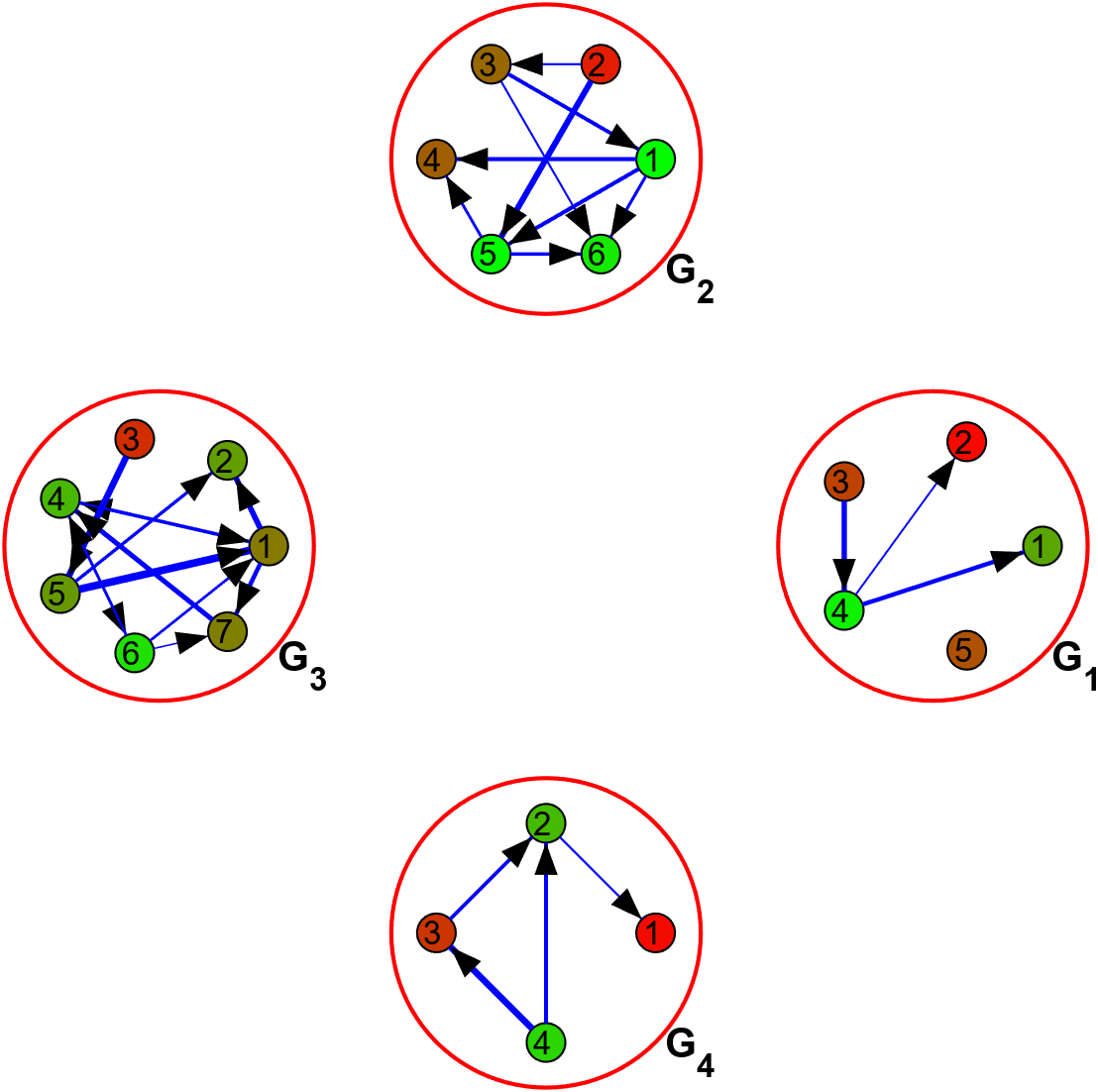}
\label{Fig:1a}}
\hfill
\subfloat[Inter-group interconnections]{
\includegraphics[width=0.46\columnwidth]{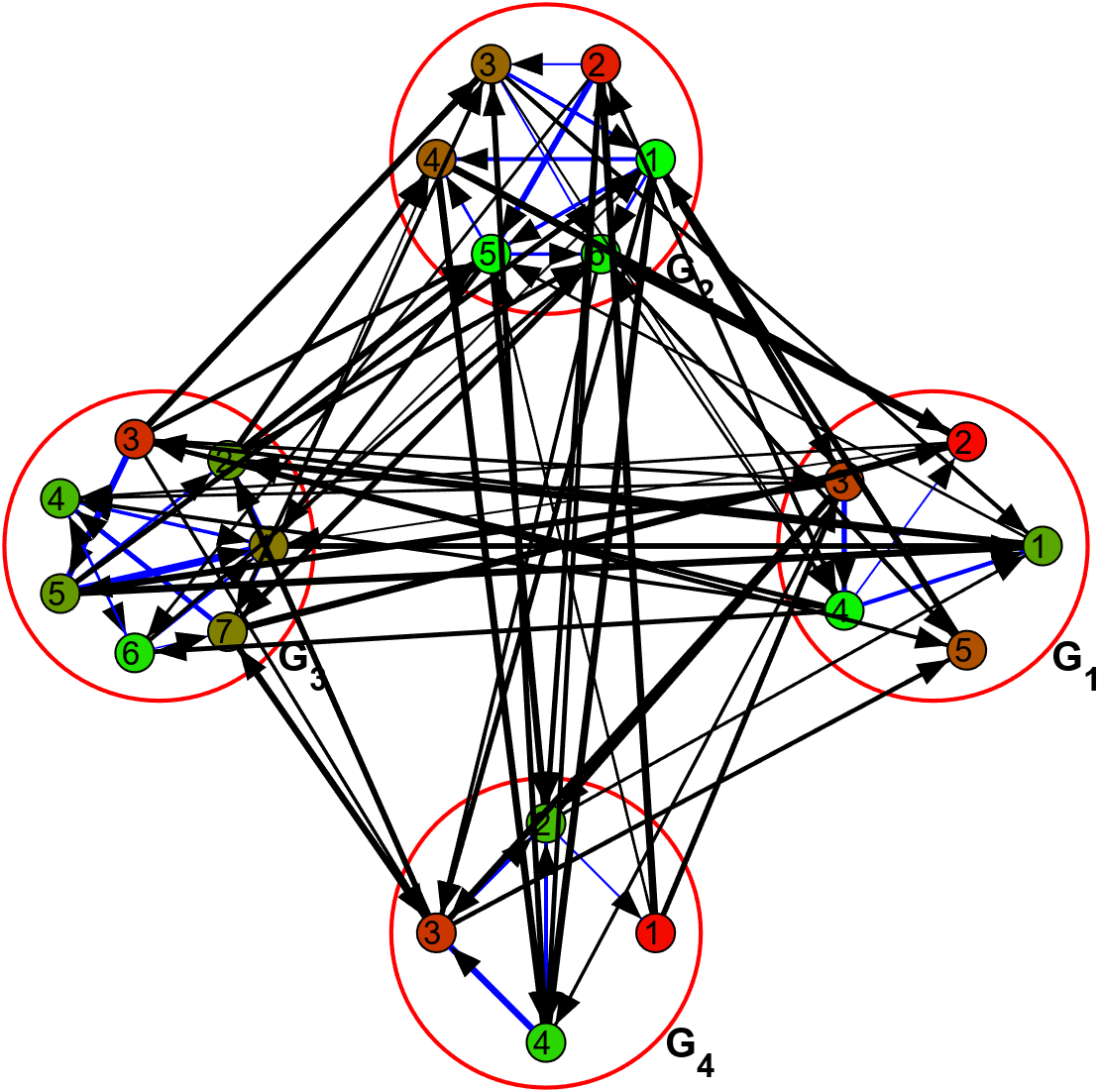}
  \label{Fig:1b}
}
\hfill
\subfloat[Dissipativity Based Design]{
\includegraphics[width=0.46\columnwidth]{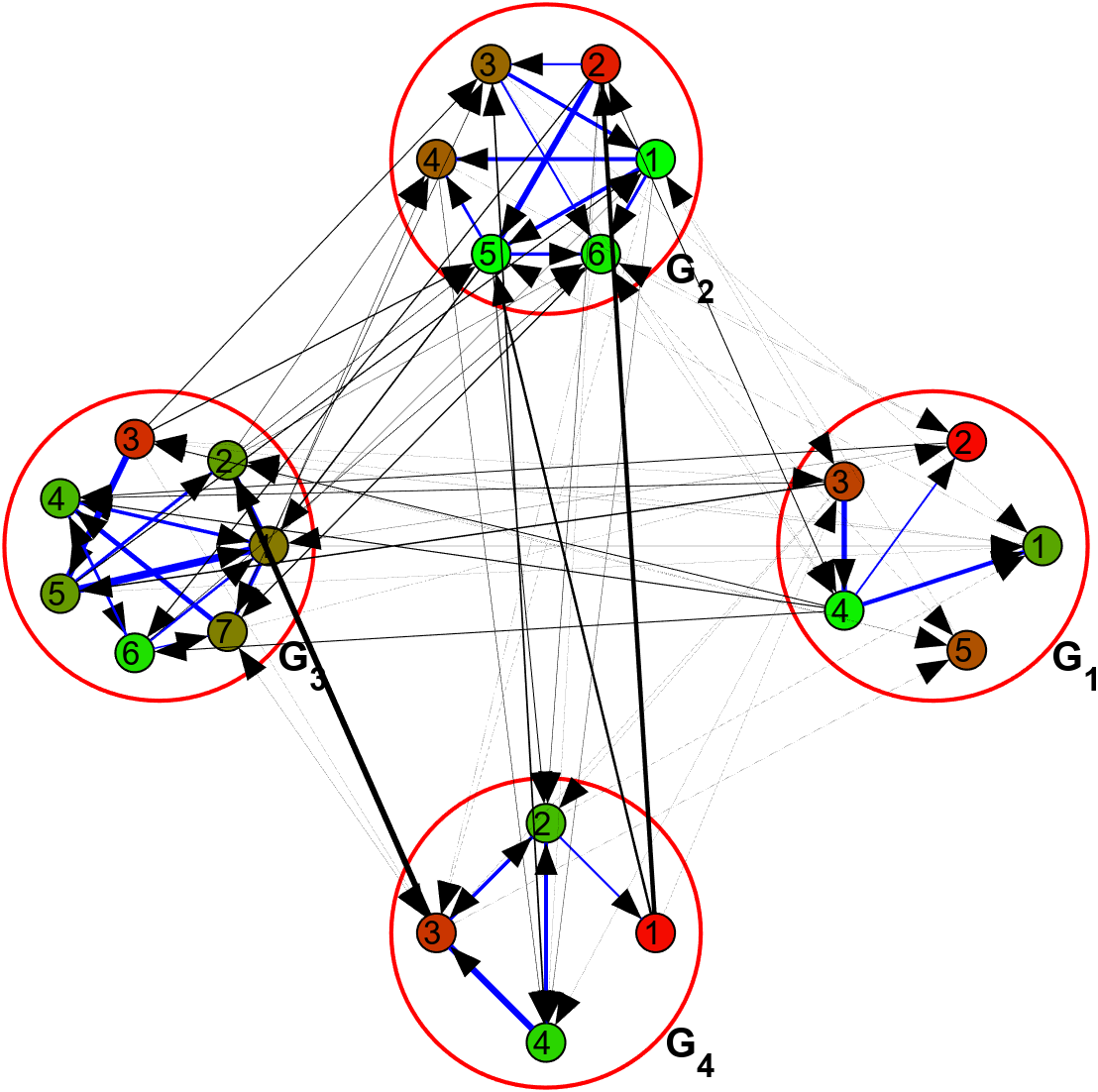}
  \label{Fig:1c}
}
\hfill
\subfloat[Interconnection Reductions]{
\includegraphics[width=0.46\columnwidth]{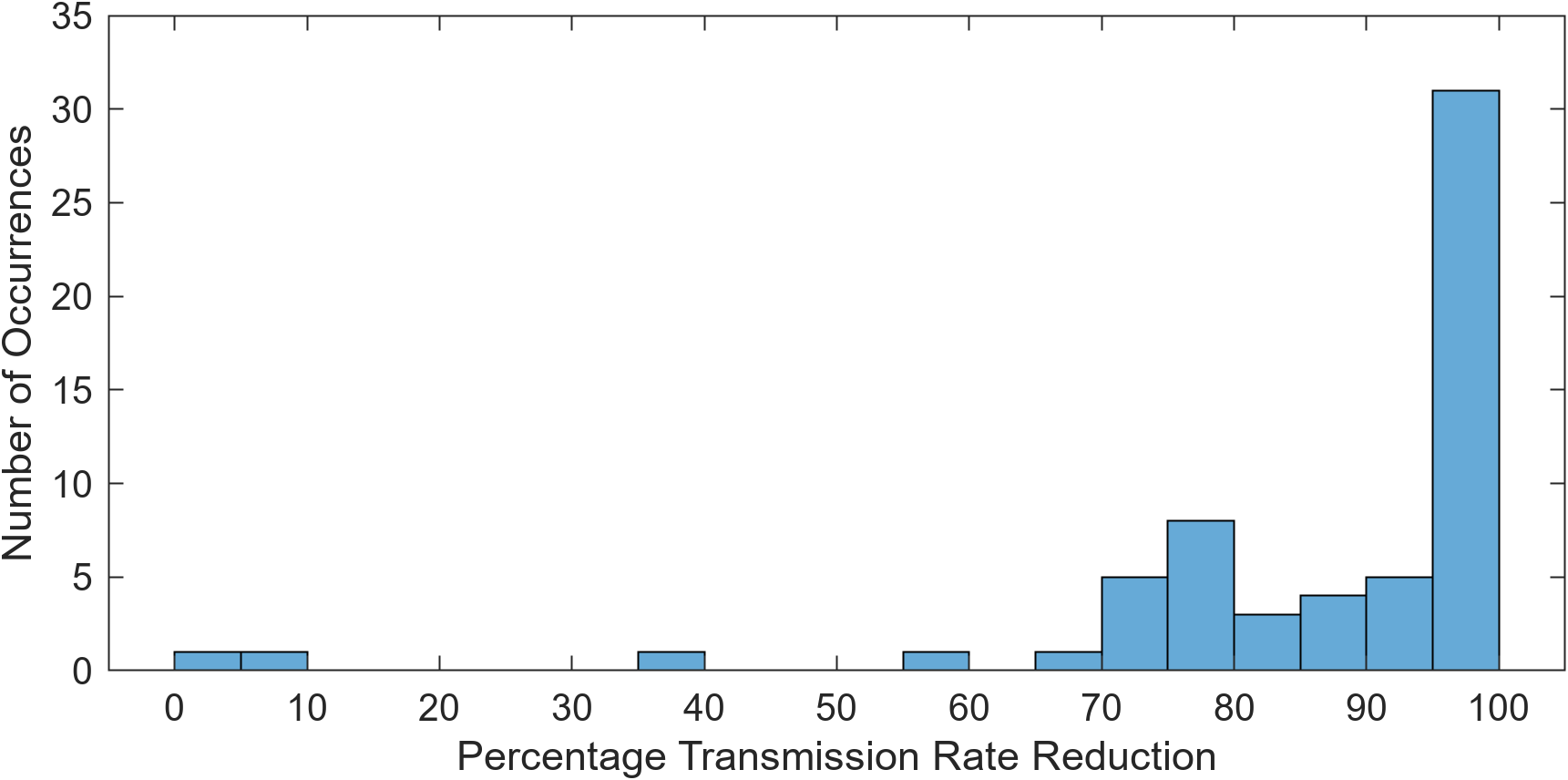}
  \label{Fig:1d}
}
\caption{Considered spreading network: (a) without and (b) with inter-group interconnections. Proposed dissipativity-based design DissBC($1,1$): (c) interconnection topology and (d) histogram of percentage interconnection reductions.}
\label{Fig:1}
\end{figure*}

\begin{figure*}[!th]
\centering
\includegraphics[width=0.95\textwidth]{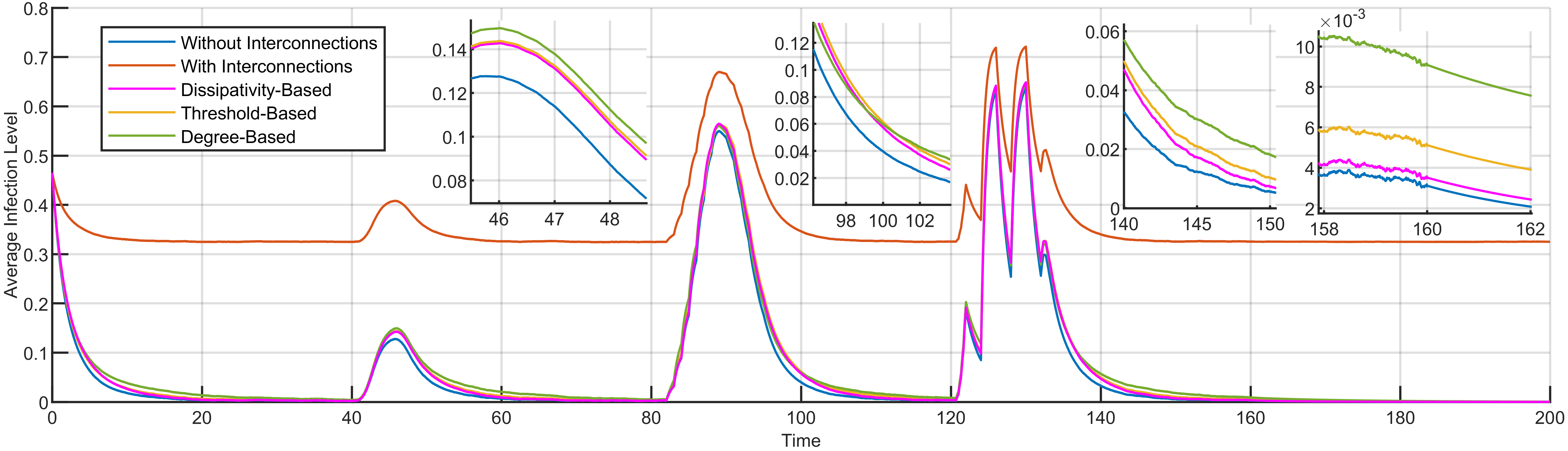}
\label{fig:2a}
\caption{Time evolution of average infection levels over the spreading network under different inter-group interconnection configurations: (1) No interconnections (blue), (2) Original uncontrolled network (orange), (3) Dissipativity-based design (magenta), (4) Threshold-based design (yellow), and (5) Degree-based design (green).}
\label{Fig:2}
\end{figure*}

\begin{figure}[!th]
\centering
% Row 1: method 9 topo + hist
\subfloat[DissBC($10^{-9},1$) Topology]{
  \includegraphics[width=0.45\columnwidth]{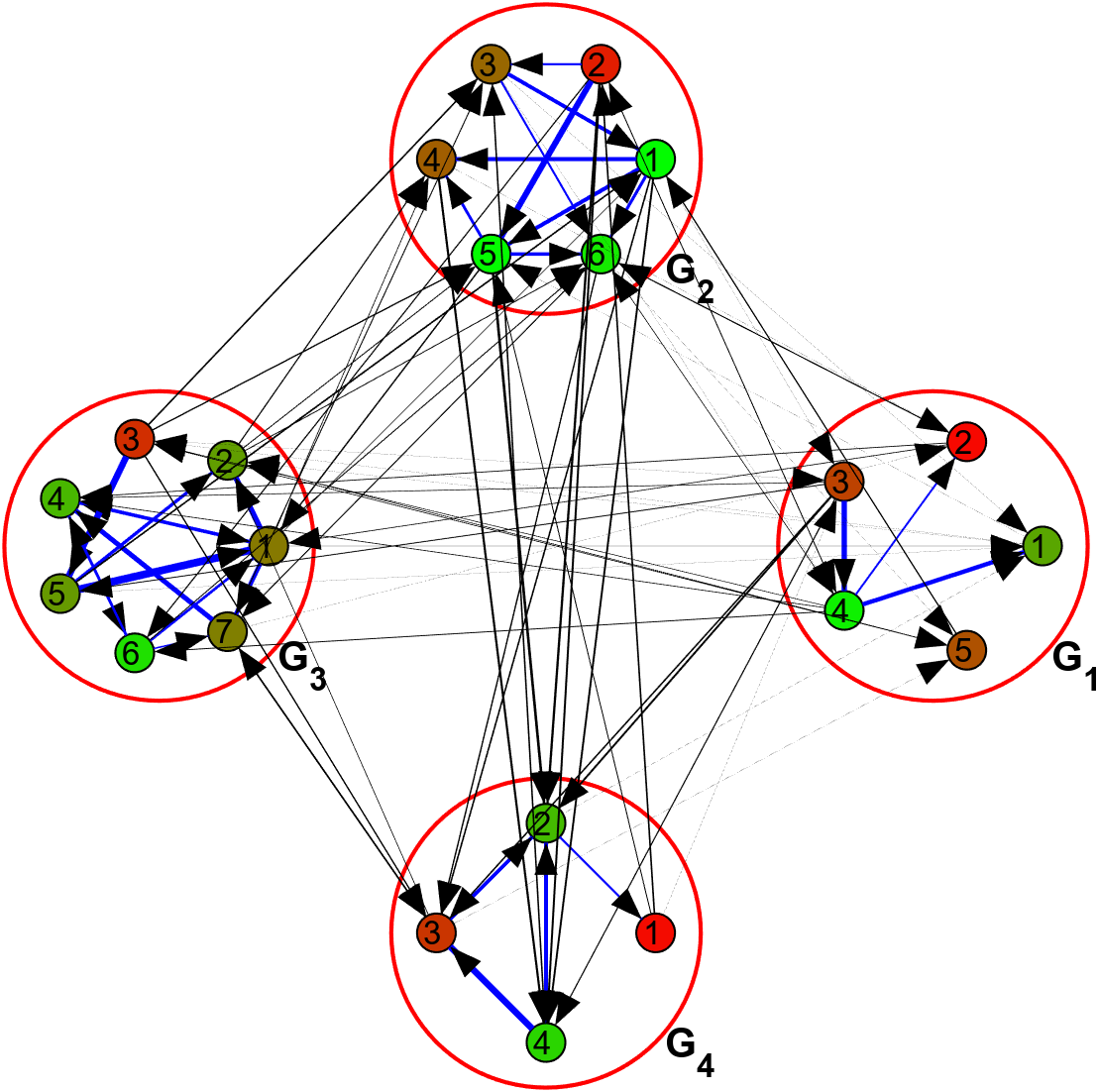}
  \label{Fig:4a}
}
\hfill
\subfloat[DissBC($10^9,1$) Topology]{
  \includegraphics[width=0.45\columnwidth]{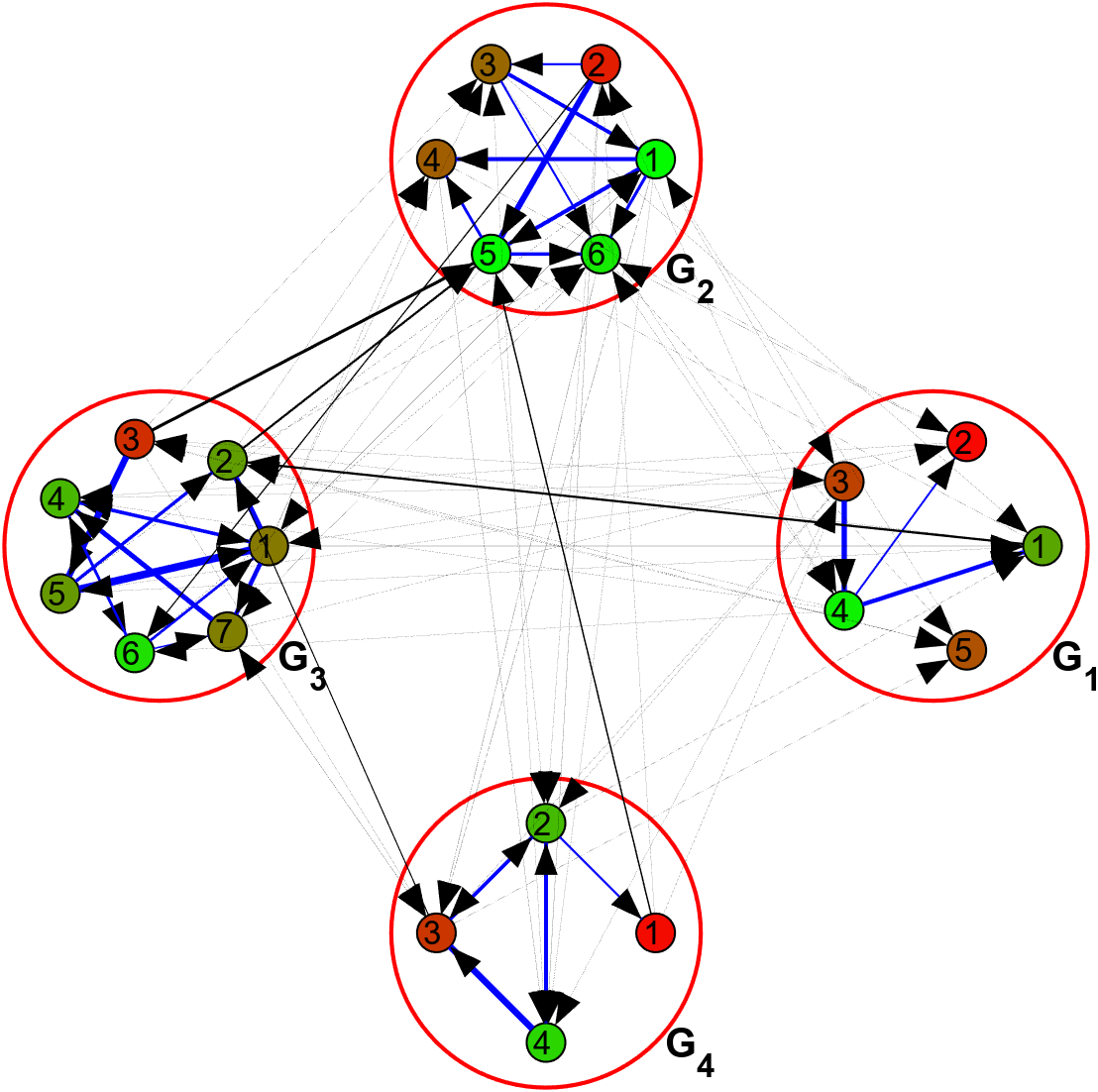}
  \label{Fig:4b}
}
\\[1ex]
\subfloat[DissBC($10^{-9},1$) Histogram]{
  \includegraphics[width=0.45\columnwidth]{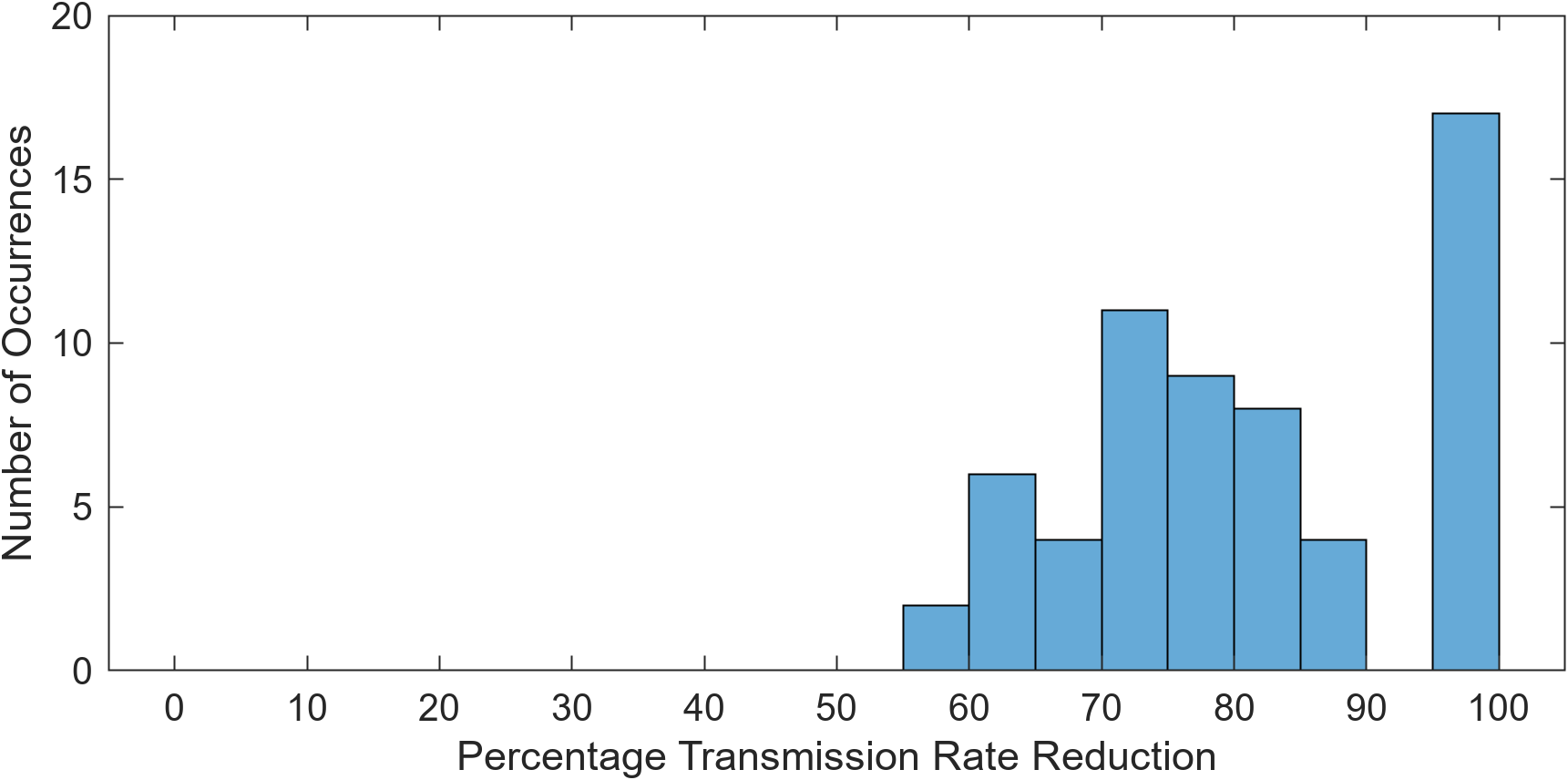}
  \label{Fig:4c}
}
\hfill
\subfloat[DissBC($10^9,1$) Histogram]{
  \includegraphics[width=0.45\columnwidth]{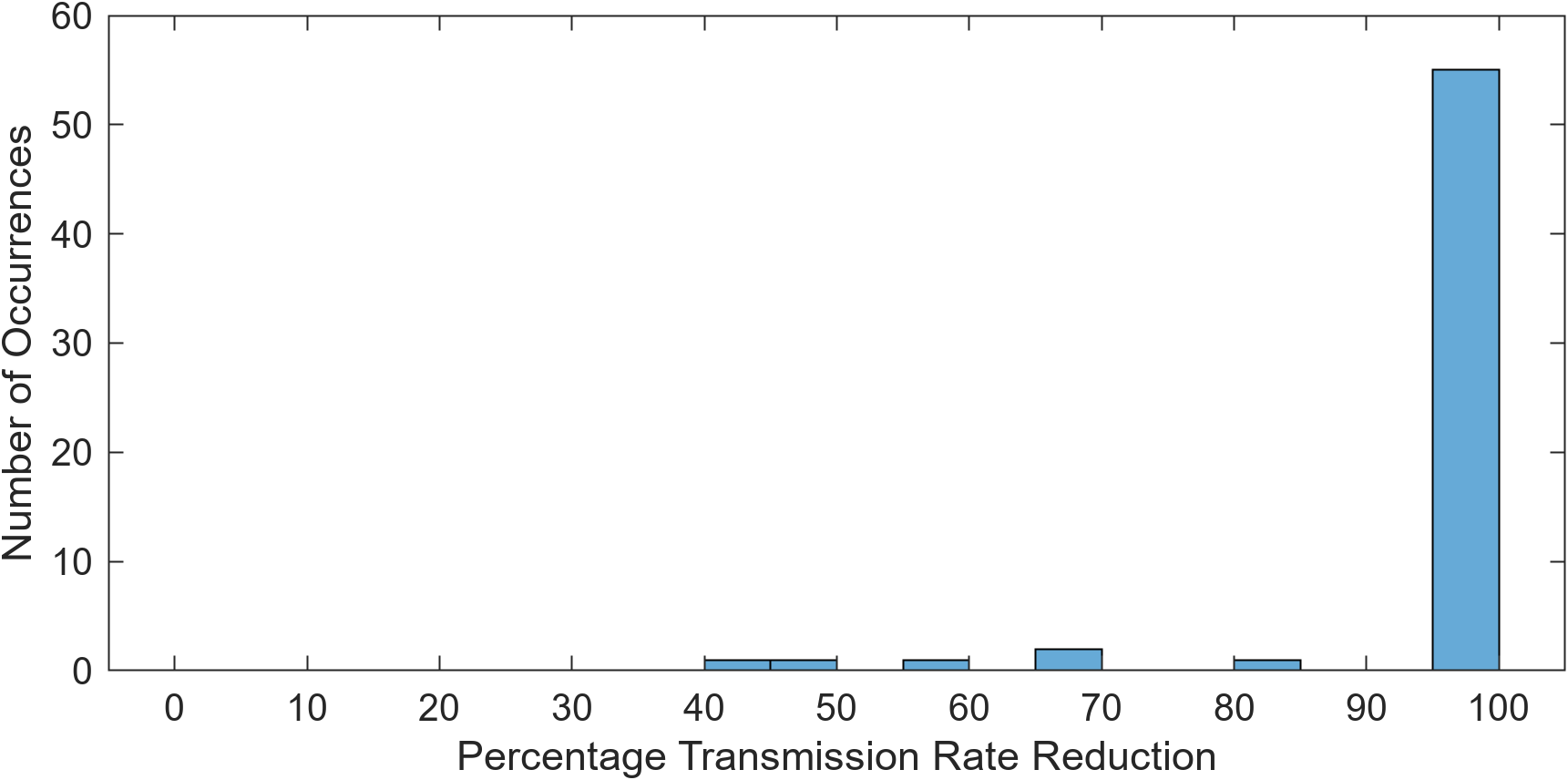}
  \label{Fig:4d}
}
\caption{Spreading network designs (top row) and their percentage interconnection reductions histograms (bottom row) observed under the proposed dissipativity-based design approach with parameters $c_M = 10^{-9}$ (left column) and $c_M = 10^{9}$ (right column) with $\delta_M = 1$.}
\label{Fig:Fig4}
\end{figure}

\begin{figure}[!th]
  \centering
  % Row 1: method 4 topo + hist
  \subfloat[DissBC($1,0.95$) Topology]{
    \includegraphics[width=0.45\columnwidth]{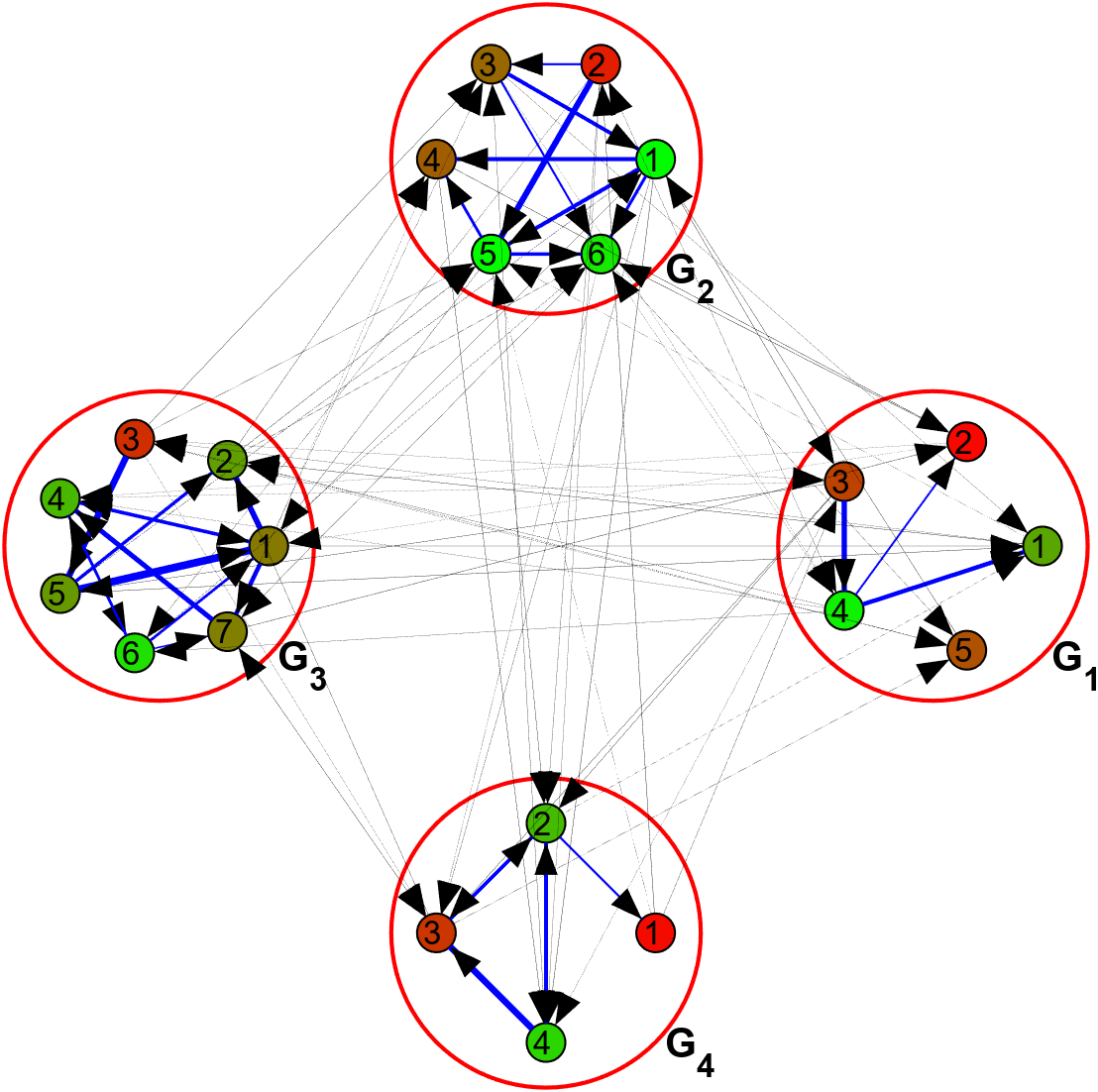}
    \label{Fig:3b}
  }
  \hfill
  \subfloat[DissBC($1,0.9$) Topology]{
    \includegraphics[width=0.45\columnwidth]{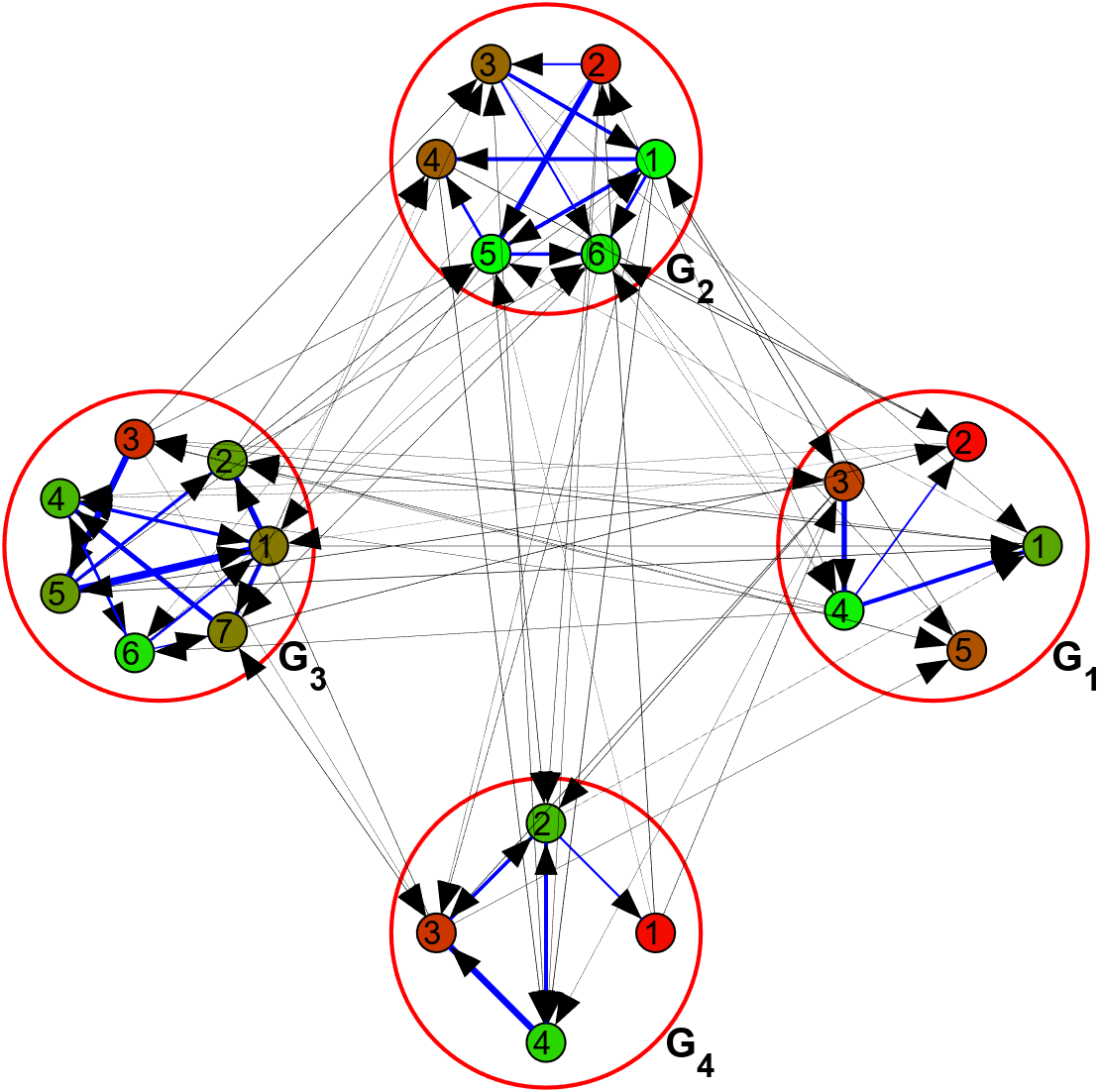}
    \label{Fig:3a}
  }
  \\[1ex]
  \subfloat[DissBC($1,0.95$) Histogram]{
    \includegraphics[width=0.45\columnwidth]{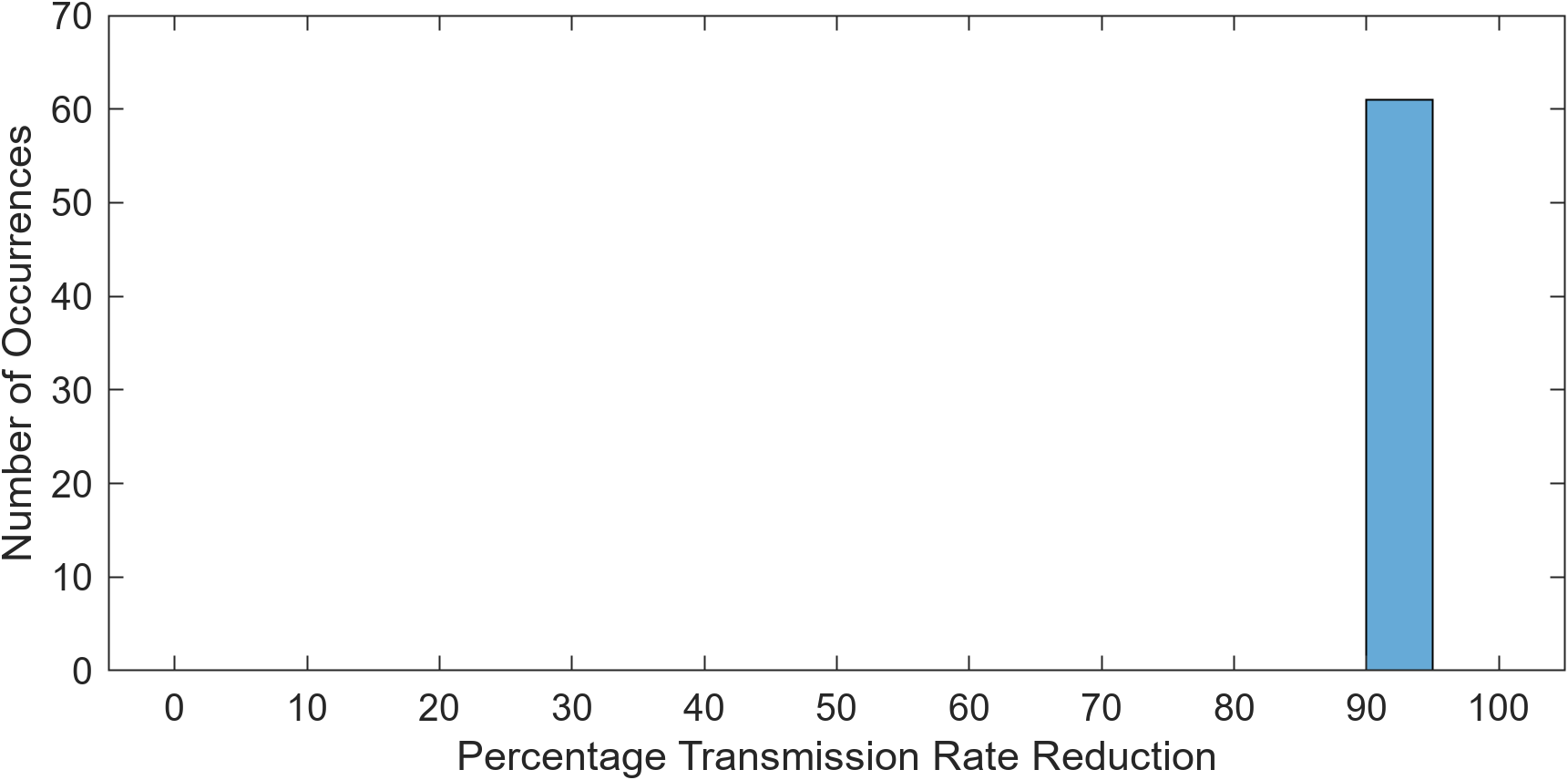}
    \label{Fig:3d}
  }
  \hfill
  \subfloat[DissBC($1,0.9$) Histogram]{
    \includegraphics[width=0.45\columnwidth]{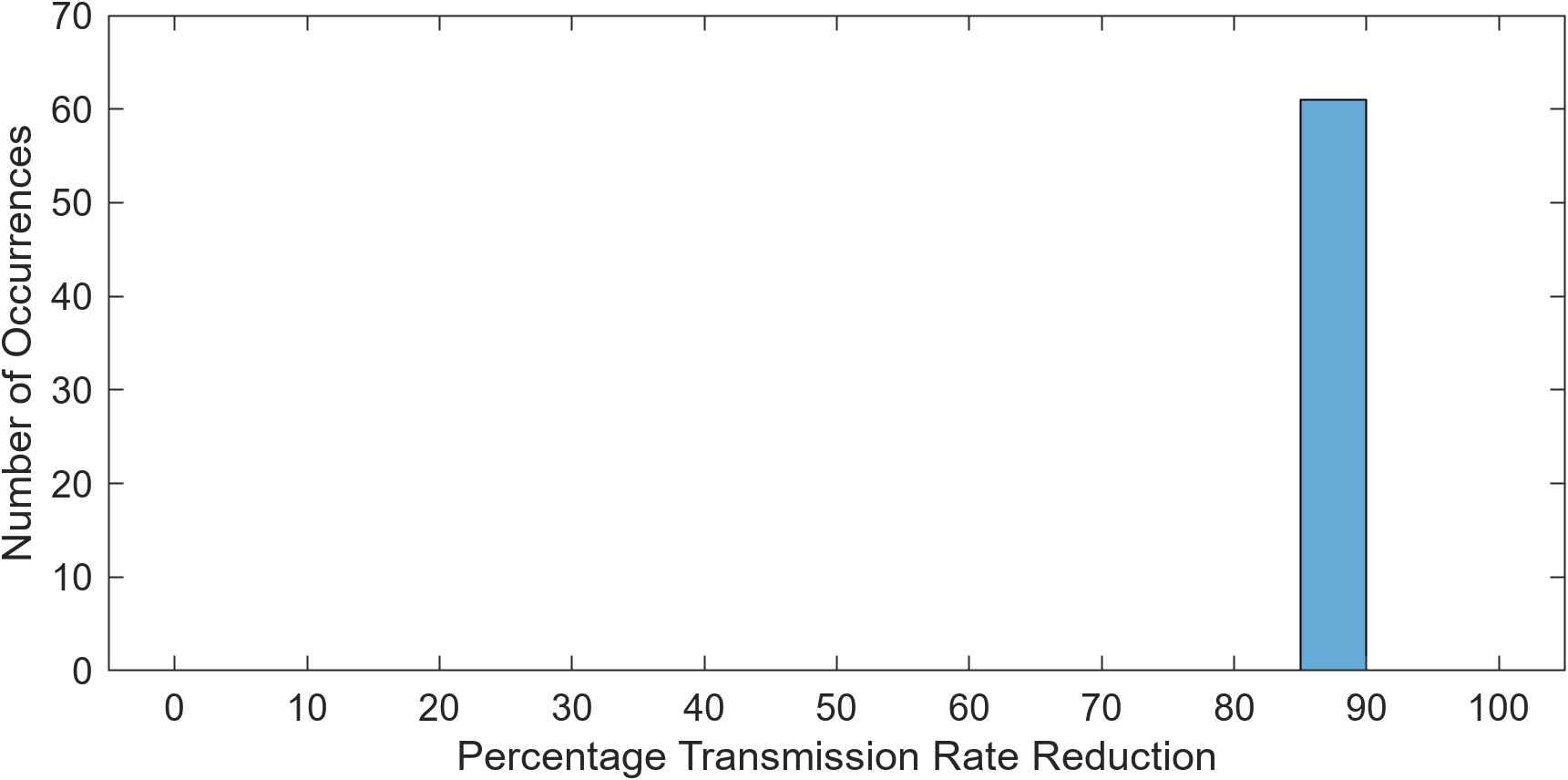}
    \label{Fig:3c}
  }
  \caption{Spreading network designs (top row) and their percentage interconnection reductions histograms (bottom row) observed under the proposed dissipativity-based design approach with parameters $c_M = 1$ with $\delta_M = 0.95$ (left column) and $\delta_M = 0.9$ (right column).}
  \label{Fig:3}
  \end{figure}

\begin{figure}[!th]
\centering
\subfloat[TBC($0.72$) Topology]{
  \includegraphics[width=0.45\columnwidth]{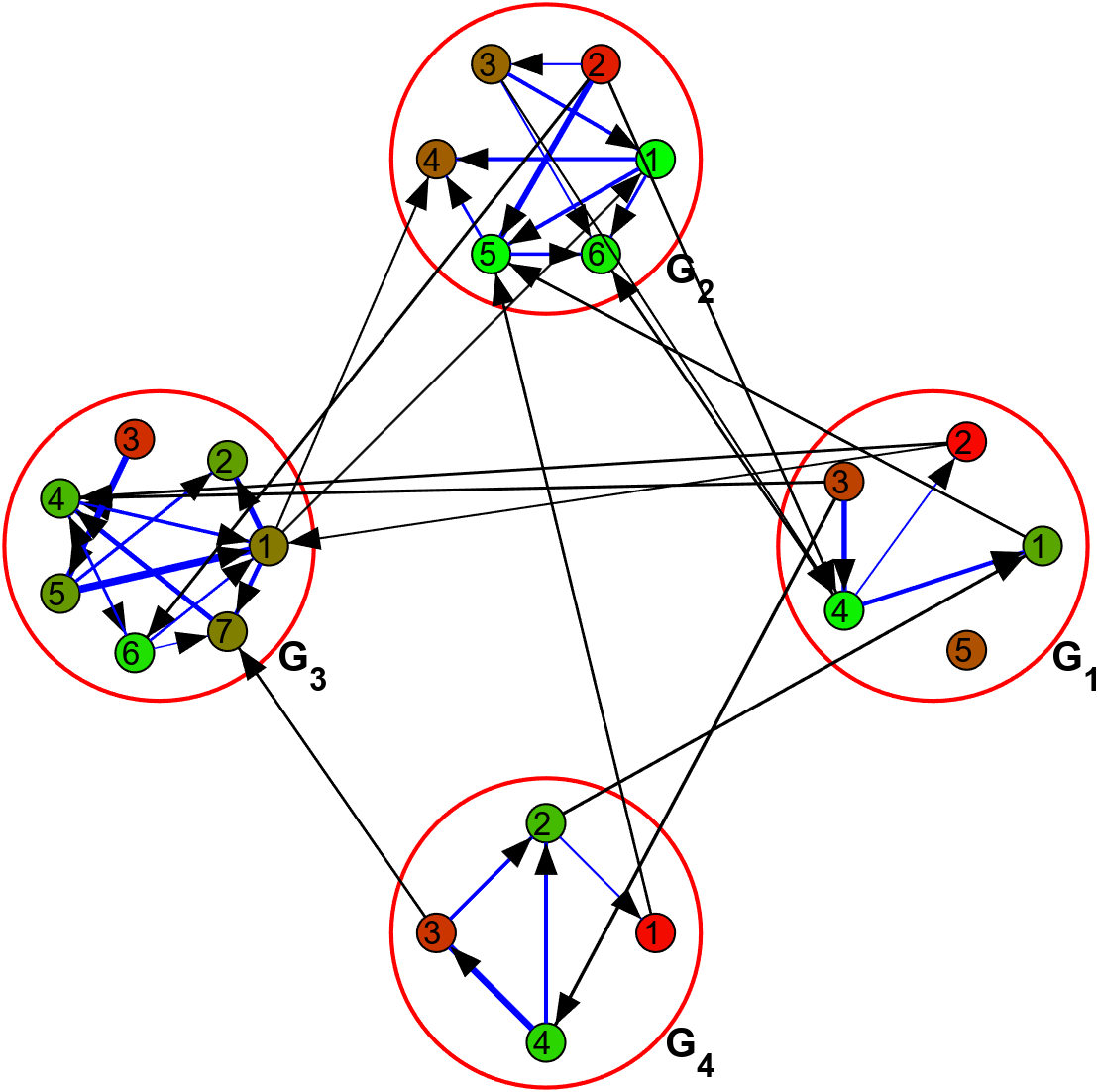}
  \label{Fig:5a}
}
\hfill
\subfloat[DegBC($0.18$) Topology]{
  \includegraphics[width=0.45\columnwidth]{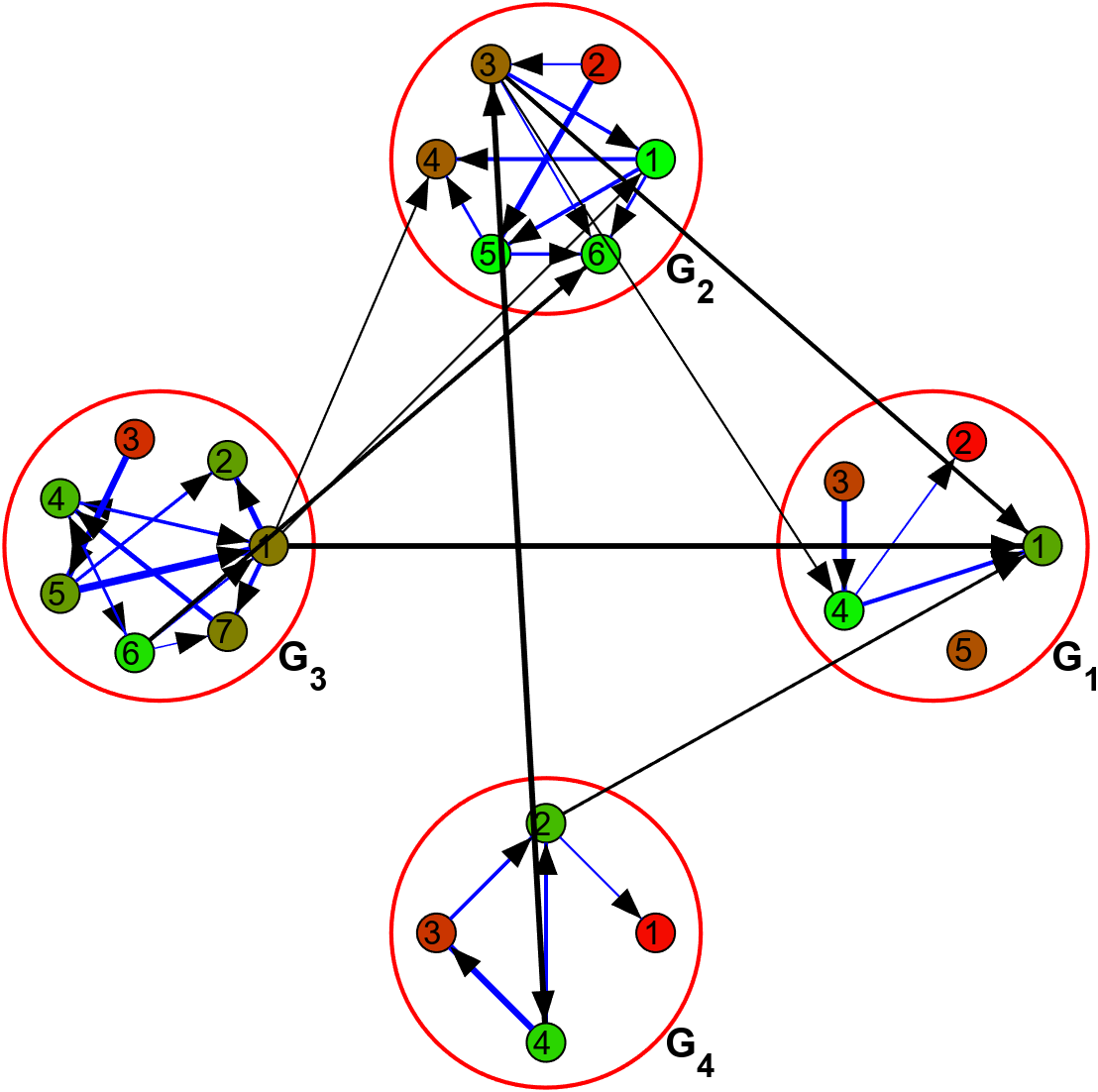}
  \label{Fig:5b}
}
\caption{Spreading network designs observed under: (a) threshold-based approach with the threshold parameter $t_M = 0.18$ and (b) degree-based approach with the degree parameter $d_M = 0.72$}
\label{Fig:5}
\end{figure}

\begin{table}[!th]
\centering
\caption{Observed performance metrics: $J_x$ (average infections), $J_M$ (average reductions), and $\gamma$ ($L_2$-gain) under different spreading network design techniques.}
\label{Tab:Results}
\begin{tabular}{| l | c | c | c | l |} 
\textbf{Design Method}    & $J_x$ & $J_M$ & $\gamma$ \\
Without interconnections  & 0.0616 & 1.0000 & -  \\
With Interconnections     & 0.3603 & 0.0000 & -  \\
DissBC($1,1$)             & 0.0680 & 0.8710 & 24.15 \\
DissBC($1,0.95$)          & 0.0645 & 0.9497 & 26.16 \\
DissBC($1,0.9$)           & 0.0677 & 0.9000 & 35.08 \\
DissBC($10^{-9},1$)       & 0.0709 & 0.8168 & 24.10 \\
DissBC($10^{9},1$)        & 0.0638 & 0.9575 & 30.70 \\
TBC($0.18$)               & 0.0700 & 0.8726 & -  \\
DegBC($0.72$)             & 0.0730 & 0.8854 & 249.6  
\end{tabular}
\end{table}

\section{Conclusion}
\label{Sec:Conclusion}
In this paper, we presented a dissipativity-based framework for controlling epidemic spreading over networked systems. By formulating the control problem using dissipativity theory, we established a systematic method for designing inter-group interconnections that ensure stability and robustness against disturbances. Comparative analysis with traditional control techniques, including threshold pruning and high-degree edge removal, demonstrates the superior performance of the proposed approach in reducing infection levels with lower control effort. Future work includes extending the framework to account for adaptive or time-varying networks and exploring its application to other spreading processes beyond epidemic control.

\bibliographystyle{IEEEtran}
\bibliography{References}

\begin{thebibliography}{10}
\providecommand{\url}[1]{#1}
\csname url@samestyle\endcsname
\providecommand{\newblock}{\relax}
\providecommand{\bibinfo}[2]{#2}
\providecommand{\BIBentrySTDinterwordspacing}{\spaceskip=0pt\relax}
\providecommand{\BIBentryALTinterwordstretchfactor}{4}
\providecommand{\BIBentryALTinterwordspacing}{\spaceskip=\fontdimen2\font plus
\BIBentryALTinterwordstretchfactor\fontdimen3\font minus \fontdimen4\font\relax}
\providecommand{\BIBforeignlanguage}[2]{{%
\expandafter\ifx\csname l@#1\endcsname\relax
\typeout{** WARNING: IEEEtran.bst: No hyphenation pattern has been}%
\typeout{** loaded for the language `#1'. Using the pattern for}%
\typeout{** the default language instead.}%
\else
\language=\csname l@#1\endcsname
\fi
#2}}
\providecommand{\BIBdecl}{\relax}
\BIBdecl

\bibitem{Hethcote2000}
H.~W. Hethcote, ``{The Mathematics of Infectious Diseases},'' \emph{SIAM Review}, vol.~42, no.~4, pp. 599--653, 2000.

\bibitem{Pastor-Satorras2001}
R.~Pastor-Satorras and A.~Vespignani, ``{Epidemic Spreading in Scale-Free Networks},'' \emph{Phys. Rev. Lett.}, vol.~86, no.~14, pp. 3200--3203, 2001.

\bibitem{Barrett2009}
C.~L. Barrett, K.~Bisset, J.~Chen, S.~Eubank, B.~Lewis, V.~S.~A. Kumar, M.~V. Marathe, and H.~S. Mortveit, \emph{{Interactions Among Human Behavior, Social Networks, and Societal Infrastructures: A Case Study in Computational Epidemiology}}.\hskip 1em plus 0.5em minus 0.4em\relax Dordrecht: Springer Netherlands, 2009, pp. 477--507.

\bibitem{Nowzari2016}
C.~Nowzari, V.~M. Preciado, and G.~J. Pappas, ``{Analysis and Control of Epidemics: A Survey of Spreading Processes on Complex Networks},'' \emph{IEEE Control Systems Magazine}, vol.~36, no.~1, pp. 26--46, 2016.

\bibitem{Sivaraman2023}
N.~K. Sivaraman, S.~Baijal, and S.~B. Muthiah, ``{On the Usage of Epidemiological Models for Information Diffusion Over Twitter},'' \emph{Social Network Analysis and Mining}, vol.~13, no.~1, p. 133, 2023.

\bibitem{Wang2024}
B.~Wang, S.~Chen, and G.~Xiao, ``{Advancing Healthcare Through Mobile Collaboration: A Survey of Intelligent Nursing Robots Research},'' \emph{Frontiers in Public Health}, vol.~12, p. 1368805, 2024.

\bibitem{Wang2024b}
W.~Wang, Y.~Nie, W.~Li, T.~Lin, M.-S. Shang, S.~Su, Y.~Tang, Y.-C. Zhang, and G.-Q. Sun, ``{Epidemic Spreading on Higher-Order Networks},'' \emph{Physics Reports}, vol. 1056, pp. 1--70, 2024.

\bibitem{Tian2022}
\BIBentryALTinterwordspacing
Y.~Tian and O.~Yagan, ``{Spreading Processes with Population Heterogeneity Over Multi-Layer Networks},'' \emph{arXiv e-prints}, p. 2211.07479, 2022. [Online]. Available: \url{http://arxiv.org/abs/2211.07479}
\BIBentrySTDinterwordspacing

\bibitem{Xue2019}
D.~Xue and S.~Hirche, ``{Distributed Topology Manipulation to Control Epidemic Spreading Over Networks},'' \emph{IEEE Trans. on Signal Processing}, vol.~67, no.~5, pp. 1163--1174, 2019.

\bibitem{Doostmohammadian2023}
M.~Doostmohammadian and H.~R. Rabiee, ``{Network-Based Control of Epidemic via Flattening the Infection Curve: High-Clustered vs. Low-Clustered Social Networks.}'' \emph{Social Network Analysis and Mining}, vol.~13, no.~1, p.~60, 2023.

\bibitem{Gou2017}
W.~Gou and Z.~Jin, ``{How Heterogeneous Susceptibility and Recovery Rates Affect the Spread of Epidemics on Networks},'' \emph{Infectious Disease Modelling}, vol.~2, no.~3, pp. 353--367, 2017.

\bibitem{Preciado2014}
C.~E. A.~J. {V. M. Preciado M. Zargham} and G.~Pappas, ``{Optimal Resource Allocation for Network Protection Against Spreading Processes},'' \emph{IEEE Transactions on Control of Network Systems}, 2014.

\bibitem{Mieghem2011}
P.~{Van Mieghem}, D.~{Stevanovi\ifmmode \acutec\else {\'{c}}\fi}, F.~Kuipers, C.~Li, R.~van~de Bovenkamp, D.~Liu, and H.~Wang, ``{Decreasing the Spectral Radius of a Graph by Link Removals},'' \emph{Phys. Rev. E}, vol.~84, no.~1, p. 16101, jul 2011.

\bibitem{Schneider2011}
C.~M. Schneider, T.~Mihaljev, S.~Havlin, and H.~J. Herrmann, ``{Suppressing Epidemics With a Limited Amount of Immunization Units},'' \emph{Phys. Rev. E}, vol.~84, no.~6, p. 61911, 2011.

\bibitem{Enns2015}
E.~A. Enns and M.~L. Brandeau, ``{Link Removal for the Control of Stochastically Evolving Epidemics Over Networks: A Comparison of Approaches},'' \emph{Journal of Theoretical Biology}, vol. 371, pp. 154--165, 2015.

\bibitem{Xu2021c}
Y.~Xu, T.~Ren, and S.~Sun, ``{Identifying Influential Edges by Node Influence Distribution and Dissimilarity Strategy},'' \emph{Mathematics}, vol.~9, no.~20, 2021.

\bibitem{Zhang2024}
\BIBentryALTinterwordspacing
H.~Zhang, W.~Tan, M.~Yu, and Y.~Li, ``{Feedback Control of the COVID-19 Outbreak Based on Active Disturbance Rejection Control},'' \emph{Applied Mathematics in Science and Engineering}, vol.~32, no.~1, p. 2325520, 2024. [Online]. Available: \url{https://doi.org/10.1080/27690911.2024.2325520}
\BIBentrySTDinterwordspacing

\bibitem{Walsh2025}
L.~Walsh, M.~Ye, B.~D.~O. Anderson, and Z.~Sun, ``{Decentralised Adaptive-Gain Control for Eliminating Epidemic Spreading on Networks},'' \emph{Automatica}, vol. 174, p. 112143, 2025.

\bibitem{Hashimoto2021}
\BIBentryALTinterwordspacing
K.~Hashimoto, Y.~Onoue, M.~Ogura, and T.~Ushio, ``{Event-Triggered Control for Mitigating SIS Spreading Processes},'' \emph{Annual Reviews in Control}, vol.~52, pp. 479--494, 2021. [Online]. Available: \url{https://www.sciencedirect.com/science/article/pii/S1367578821000699}
\BIBentrySTDinterwordspacing

\bibitem{Kovacevic2022}
\BIBentryALTinterwordspacing
R.~M. Kovacevic, N.~I. Stilianakis, and V.~M. Veliov, ``{A Distributed Optimal Control Model Applied to COVID-19 Pandemic},'' \emph{SIAM Journal on Control and Optimization}, vol.~60, no.~2, pp. S221--S245, 2022. [Online]. Available: \url{https://doi.org/10.1137/20M1373840}
\BIBentrySTDinterwordspacing

\bibitem{Carli2020}
\BIBentryALTinterwordspacing
R.~Carli, G.~Cavone, N.~Epicoco, P.~Scarabaggio, and M.~Dotoli, ``{Model Predictive Control To Mitigate the COVID-19 Outbreak in a Multi-Region Scenario},'' \emph{Annual Reviews in Control}, vol.~50, p. 373, 2020. [Online]. Available: \url{https://pmc.ncbi.nlm.nih.gov/articles/PMC7528763/}
\BIBentrySTDinterwordspacing

\bibitem{SheHale2024}
B.~She and M.~Hale, ``{A Dissipativity Approach to Analyzing Composite Spreading Networks},'' \emph{arXiv preprint arXiv:2412.02665}, 2024.

\bibitem{WelikalaP52022}
S.~Welikala, H.~Lin, and P.~J. Antsaklis, ``{Non-Linear Networked Systems Analysis and Synthesis using Dissipativity Theory},'' in \emph{Proc. of American Control Conf.}, 2023, pp. 2951--2956.

\bibitem{WelikalaP72023}
S.~Welikala, Z.~Song, H.~Lin, and P.~J. Antsaklis, ``{Decentralized Co-Design of Distributed Controllers and Communication Topologies for Vehicular Platoons: A Dissipativity-Based Approach},'' \emph{Automatica}, vol. 174, no. 0005-1098, p. 112118, 2025.

\bibitem{Welikala2025Ax1}
\BIBentryALTinterwordspacing
S.~Welikala, H.~Lin, and P.~J. Antsaklis, ``{Inventory Consensus Control in Supply Chain Networks using Dissipativity-Based Control and Topology Co-Design},'' \emph{arXiv e-prints}, p. 2502.06580, 2025. [Online]. Available: \url{https://arxiv.org/abs/2502.06580}
\BIBentrySTDinterwordspacing

\bibitem{Najafirad2025P1}
M.~J. Najafirad and S.~Welikala, ``{Distributed Dissipativity-Based Controller and Topology Co-Design for DC Microgrids},'' in \emph{Proc. of American Control Conf. (accepted)}, 2025.

\bibitem{Mirabilio2022}
M.~Mirabilio, A.~Iovine, E.~{De Santis}, M.~D.~D. Benedetto, and G.~Pola, ``{Scalable Mesh Stability of Nonlinear Interconnected Systems},'' \emph{IEEE Control Systems Letters}, vol.~6, pp. 968--973, 2022.

\bibitem{Lofberg2004}
J.~Lofberg, ``{YALMIP : A Toolbox for Modeling and Optimization in MATLAB},'' in \emph{Proc. of IEEE Intl. Conf. on Robotics and Automation}, 2004, pp. 284--289.

\bibitem{Sontag1995}
E.~D. Sontag and Y.~Wang, ``{On Characterizations of the Input-to-State Stability Property},'' \emph{Systems \& Control Letters}, vol.~24, no.~5, pp. 351--359, 1995.

\bibitem{Willems1972a}
J.~C. Willems, ``{Dissipative Dynamical Systems Part I: General Theory},'' \emph{Archive for Rational Mechanics and Analysis}, vol.~45, no.~5, pp. 321--351, 1972.

\bibitem{Kottenstette2014}
N.~Kottenstette, M.~J. McCourt, M.~Xia, V.~Gupta, and P.~J. Antsaklis, ``{On Relationships Among Passivity, Positive Realness, and Dissipativity in Linear Systems},'' \emph{Automatica}, vol.~50, no.~4, pp. 1003--1016, 2014.

\bibitem{WelikalaJ22022}
S.~Welikala, H.~Lin, and P.~J. Antsaklis, ``{A Decentralized Analysis and Control Synthesis Approach for Networked Systems with Arbitrary Interconnections},'' \emph{IEEE Trans. on Automatic Control}, no. 0018-9286, 2024.

\bibitem{VanAntwerp2000}
J.~G. VanAntwerp and R.~D. Braatz, ``{A Tutorial on Linear and Bilinear Matrix Inequalities},'' \emph{Journal of Process Control}, vol.~10, no.~4, pp. 363--385, 2000.

\bibitem{Simon2011}
E.~Simon, P.~R-Ayerbe, C.~Stoica, D.~Dumur, and V.~Wertz, ``{LMIs-Based Coordinate Descent Method for Solving BMIs in Control Design},'' \emph{IFAC Proceedings Volumes}, vol.~44, no.~1, pp. 10\,180--10\,186, 2011.

\bibitem{Lee2019}
\BIBentryALTinterwordspacing
D.~Lee and J.~Hu, ``{Sequential Parametric Convex Approximation Algorithm for Bilinear Matrix Inequality Problem},'' \emph{Optimization Letters}, vol.~13, no.~4, pp. 741--759, 2019. [Online]. Available: \url{https://doi.org/10.1007/s11590-018-1274-6}
\BIBentrySTDinterwordspacing

\bibitem{Yi2022}
\BIBentryALTinterwordspacing
Y.~Yi, L.~Shan, P.~E. Par{\'{e}}, and K.~H. Johansson, ``{Edge Deletion Algorithms for Minimizing Spread in SIR Epidemic Models},'' \emph{arXiv e-prints}, p. 2011.11087, 2020. [Online]. Available: \url{http://arxiv.org/abs/2011.11087}
\BIBentrySTDinterwordspacing

\end{thebibliography}

\end{document}